\documentclass[11pt]{article}
\usepackage[a4paper,margin=1in]{geometry}
\usepackage{amsthm,amssymb,graphicx,latexsym,amsfonts,color,lscape}
\PassOptionsToPackage{hyphens}{url}\usepackage{hyperref}

\newcommand{\multiSM}{\textsc{Stable Multigraph b-Matching}}
\newcommand{\fixt}{\textsc{sf}}
\newcommand{\hr}{\textsc{hr}}
\newcommand{\hrc}{\textsc{hrc}}

\newtheorem{lemma}{Lemma}
\newtheorem{claim}[lemma]{Claim}
\newtheorem{remark}[lemma]{Remark}
\newtheorem{obs}[lemma]{Observation}
\newtheorem{ex}[lemma]{Example}
\newtheorem{definition}[lemma]{Definition}

\newtheorem{theorem}[lemma]{Theorem}
\newtheorem{corollary}[lemma]{Corollary}
\newcommand{\resp}{sub-responsive}
\newcommand{\hrcdual}{\textsc{hrc-dual market}}
\newcommand{\subcomp}{sub-complete}

\newcommand{\sat}{\small (2,2)\normalsize {\sc -e}\small 3\normalsize {\sc -sat}}
\newcommand{\claimproofstart}{\begin{proof}}
\newcommand{\claimproofend}{\renewcommand{\qedsymbol}{\scalebox{1.6}{$\diamond$}}\end{proof}}
\newcommand{\specialclaimproofstart}{\noindent\textit{Proof.\ }}
\newcommand{\specialclaimproofend}{\renewcommand{\qedsymbol}{\scalebox{1.6}{$\diamond$}}\qedhere~~\renewcommand{\qedsymbol}{$\square$}\qedsymbol}
\newcommand{\mycomment}[1]{}
\begin{document}
\title{Couples can be tractable: New algorithms and hardness \\ results for the Hospitals / Residents problem with Couples\thanks{A preliminary version of this paper appeared in the Proceedings of IJCAI 2024 \cite{CMMT24}.}}
{
\author{Gergely Csáji$^{1,2}$, David Manlove$^3$, Iain McBride$^3$ and James Trimble$^3$
\\
\vspace{-2mm}
\\
\small
$^1$ \emph{ELTE Eötvös Loránd University, Budapest, Hungary}
\\
\small
$^2$ \emph{HUN-REN KRTK KTI, Budapest, Hungary}
\\
\small
{\tt csaji.gergely@krtk.hu}.
\\
\vspace{-2mm}
\\
\small
$^2$ \emph{School of Computing Science, University of Glasgow, Glasgow G12 8QQ, UK.}
\\
\small
{\tt david.manlove@glasgow.ac.uk}, {\tt iainmcbride@hotmail.com}, {\tt james.trimble@yahoo.co.uk}.}
}
\date{ }
\maketitle

\begin{abstract}
In this paper, we study the {\sc Hospitals / Residents problem with Couples} (\hrc), where a solution is a stable matching or a report that none exists. We present a novel polynomial-time algorithm that can find a near-feasible stable matching (adjusting the hospitals' capacities by at most 1) in an \hrc\ instance where the couples' preferences are \resp\ (i.e., if one member switches to a better hospital, then the couple also improves) and \subcomp\ (i.e., each pair of hospitals that are individually acceptable to both members are jointly acceptable for the couple) by reducing it to an instance of the {\sc Stable Fixtures} problem.  We also present a polynomial-time algorithm for \hrc\ in a \resp, \subcomp\ instance that is a Dual Market, or where all couples are one of several possible types. We show that our algorithm also implies the polynomial-time solvability of a stable b-matching problem, where the underlying graph is a multigraph with loops.

We complement our algorithms with several hardness results. We show that \hrc\ with \resp\ and \subcomp\ couples is NP-hard, even with other strong restrictions. We also show that \hrc\ with a Dual Market is NP-hard under several simultaneous restrictions.
Finally, we show that the problem of finding a matching with the minimum number of blocking pairs in \hrc\ is not approximable within $m^{1-\varepsilon}$, for any $\varepsilon>0$, where $m$ is the total length of the hospitals' preference lists, unless P=NP, even if each couple applies to only one pair of hospitals.  

Our polynomial-time solvability results greatly expand the class of known tractable instances of \hrc\ and provide a useful tool for designing better and more efficient mechanisms in the future.
\end{abstract}

\section{Introduction}
\label{sec:intro}
\subsection{Background}
The {\sc Stable Marriage problem}, and its many-to-one extension, the {\sc Hospitals / Residents problem} (\hr), are well-studied and central problems in Computational Social Choice, Computer Science, Game Theory and Economics, with numerous applications including in entry-level labour markets, school choice and higher education allocation \cite{Man13}.  In the medical sphere, centralised matching schemes that assign aspiring junior doctors to hospitals operate in many countries.  One of the largest and best known examples is the National Resident Matching Program (NRMP) in the US \cite{ZZZ5}, which had just under 45,000 applicants for just over 41,500 positions in the Main Residency Match in 2024 \cite{ZZZ5a}. There are analogous matching schemes for junior doctor allocation in Canada \cite{ZZZ6}, France \cite{ZZZ6a} and Japan \cite{ZZZ6b}.

The \hr\ problem model represents a bipartite matching market with two-sided preferences, involving the preferences of doctors over acceptable hospitals, and those of hospitals over their applicants.  Each hospital has a capacity, indicating the maximum number of doctors that it can admit. Roth \cite{Rot84} argued that a key property to be satisfied by a matching $M$ in an instance $I$ of \hr\ is \emph{stability}, which ensures that there is no \emph{blocking pair}, comprising a doctor and a hospital, both of whom have an incentive to deviate from their assignments in $M$ and become matched to one another, undermining the integrity of $M$.  It is known that every instance of \hr\ admits a stable matching, which may be found in time linear in the size of $I$ \cite{GS62}.

\hr\ is also known to have very nice mathematical structure. In particular, the famous Rural Hospitals Theorem states that, for an arbitrary \hr\ instance $I$, any doctor that is assigned in one stable matching in $I$ is assigned in all stable matchings in $I$, moreover any hospital that is undersubscribed in some stable matching in $I$ is assigned exactly the same set of doctors in every stable matching in $I$ \cite{Rot84,GS85,Rot86}.

In many of the above applications, there may be couples amongst the applying doctors, who wish to be allocated to hospitals that are geographically close to each other, for example.  Indeed, the NRMP matching algorithm was redesigned in 1983 specifically to allow couples to provide preferences over pairs of hospitals, with each pair representing a simultaneous assignment that is suitable for both members of the couple.  We thus obtain a generalisation of \hr\ called the {\sc Hospitals / Residents problem with Couples} (\hrc).  By modifying the definition of stability, taking into account how a couple could improve relative to a matching, Roth \cite{Rot84} showed that the addition of couples destroys the crucial property that a stable matching must always exist.  In \hrc, the problem therefore is to find a stable matching or report that none exists.  Even when a stable matching does exist, such matchings may have different sizes \cite{AC96}, thus showing that the Rural Hospitals Theorem does not hold in general for \hrc.  Worse still, and again in contrast to the case for \hr, Ronn \cite{Ron90} showed that \hrc\ is NP-hard.

Since then, further NP-hardness results for \hrc\ have been established, even in very restrictive cases \cite{NH88,Ron90,biro2014hospitals}, suggesting that \hrc\ is a very challenging computational problem in theory.  Even so, the algorithm employed by the NRMP has succeeded in producing a stable matching every year (allowing for some small capacity adjustments for hospitals that may sometimes be necessary) since couples were first included in 1983.  The NRMP employs a heuristic that is based on the Roth-Peranson algorithm described in \cite{RP99}.

The success of the NRMP (and that of similar schemes) is demonstrated by its longevity, celebrating its 70th anniversary in 2022, and its high level of participation, growing year on year to the current numbers reported above.  Indeed, the contribution of Roth and Shapley to the theory of stable matching and the practice of centralised matching programmes such as the NRMP led to their award of the Nobel Prize in Economic Sciences in 2012 \cite{ZZZ6c}.

Nevertheless, up until now there have been very few special cases of \hrc\ that have been shown to admit polynomial-time algorithms.  Our work in this paper extends the techniques that can be used in \hrc\ instances by providing novel polynomial-time algorithms for many practically important variants of the problem in a surprisingly wide range of settings, and thus contributes new ways to efficiently handle specific instances in several applications. 

\subsection{Related work}

Roth \cite{Rot84} considered stability in the \hrc\ context although did not define the concept explicitly. Whilst Gusfield and Irving \cite{GI89} defined stability formally in \hrc, their definition neglected to deal with the case that both members of a couple may wish to be assigned to the same hospital simultaneously.  There are discussions of this point in \cite{BK13} and \cite[Section 5.3.2]{Man13}.  McDermid and Manlove \cite{MM10} provided a stability definition for \hrc\ that does deal with this possibility.  Other such definitions have been proposed in the literature \cite{BIS11,KPR13}, and a detailed comparative study of the effects of different stability definitions for \hrc\ was carried out by Delorme et al.\ \cite{DGGKMP21}.

In this paper, we will adopt the stability definition of McDermid and Manlove \cite{MM10}.  With respect to this definition, several algorithmic results for \hrc\ hold.  Firstly, Ronn's NP-hardness result for \hrc\ holds even in the case that each hospital has capacity 1 and there are no single doctors \cite{Ron90} (note that all \hrc\ stability definitions are equivalent in the unit hospital capacity case, or in the case that no couple wish to be assigned to the same hospital).  Ng and Hirschberg \cite{NH88} (reported in \cite{NH91}) showed that NP-hardness also holds additionally even when each couple finds acceptable every distinct pair of hospitals.

Ng and Hirschberg \cite{NH88} also proved NP-hardness for a ``dual market'' restriction of \hrc, which we refer to as \hrcdual, where the two sets $H_1$ and $H_2$, comprising the hospitals appearing in the first (respectively second) positions of the acceptable pairs of the first (respectively second) members of each couple are disjoint. This reflects the ``two body'' problem that often exists for couples, where, for example, one member of a couple may be applying for a job at a university, whilst the other couple member may be applying for a job in industry.  Another motivating application arises in a student allocation scheme, where students apply to both a university degree programme and to an internship at a company. This can also be viewed as a specific \hrcdual\ instance, where $H_1$ comprises the university degree programmes, $H_2$ comprises the companies, and each student corresponds to a couple with one member applying to universities and the other to internships.  Ng and Hirschberg's NP-hardness result for \hrcdual\ holds even if each hospital has capacity 1 and there are no single doctors.

McDermid and Manlove \cite{MM10} showed that \hrc\ is NP-hard even if each single doctor ranks at most three hospitals, each couple ranks at most two pairs of hospitals, each hospital ranks at most four doctors, and each hospital has a capacity of 1 or 2.  Later, Bir\'o et al.\ \cite{biro2014hospitals} strengthened this result by showing that NP-hardness holds for \hrc\ even when there are no single doctors, each couple ranks at most two pairs of hospitals, each hospital ranks at most two doctors, and all hospitals have capacity 1.  On the other hand, if each hospital finds at most one doctor acceptable, and there is no restriction on the lengths of the preference lists of single doctors and couples, Manlove et al.\ \cite{MMT17} showed that this special case of \hrc\ admits a unique stable matching, which may be found in polynomial time.  They also showed that \hrc\ is solvable in polynomial time when each single doctor and hospital has a preference list of length at most two, and each couple finds only one hospital pair acceptable. 
Another simple tractable case of \hrc\ is given by Klaus and Klijn \cite{klaus2005stable} (see also \cite{KKN09}), but also in a very restricted framework.  They required that 
each couple's preference list must be \emph{weakly responsive} (informally, a couple's preferences are weakly responsive if there are underlying preferences over individual hospitals for each member of the couple, such that if one member of the couple goes to a better hospital then the couple also improves), each couple must find acceptable all possible outcomes where at least one member is matched, and each hospital has capacity 1.

Given the prevalence of NP-hardness results for \hrc, and the scarcity of polynomial-time algorithms, heuristics have been applied to the problem (see \cite[Section 5.3,3]{Man13} for a survey) as well as approaches based on parameterised complexity and local search \cite{MS11,BIS11},  integer programming \cite{biro2014hospitals,DGGKMP21}, constraint programming \cite{MMT17} and SAT solving \cite{DPB15}.

Nguyen and Vohra \cite{nguyen2018near} studied so-called \emph{near-feasible} stable matchings in \hrc. They showed that if one can modify the capacities of the hospitals by at most 2, then a stable matching with respect to the new capacities always exists (the new total hospital capacity is at least as large as before and increases by at most 4). They also provided an algorithm to find such a near-feasible solution, however their algorithm is not guaranteed to run in polynomial time, as in the first step it computes a stable fractional matching using Scarf's algorithm \cite{Sca67}, and the computation of stable fractional matchings is known to be PPAD-hard \cite{csaji2022complexity}. Bir\'o et al.\ \cite{biro2016matching} also used Scarf's algorithm to find stable matchings in \hrc\ instances. 

Another direction, to cope with the possible non-existence of a stable matching, is to consider matchings that are ``as stable as possible'', i.e., admit the minimum number of blocking pairs.  We refer to this problem as {\sc min bp hrc}.  Bir\'o et al.\ \cite{biro2014hospitals} showed that this problem is NP-hard and not approximable within $n_D^{1-\varepsilon}$, for any $\varepsilon>0$, unless P=NP, where $n_D$ is the number of doctors in a given instance.  
Manlove et al.\  \cite{MMT17} presented integer and constraint programming formulations for {\sc min bp hrc}.

\subsection{Our contributions}
In this paper, we provide new polynomial-time algorithms for \hrc\ in the case that the couples' preference lists are \emph{\subcomp} and \emph{\resp}.  Informally, a couple's preferences are \subcomp\ if there are underlying preferences for the couple members, such that if both members go to a hospital that is acceptable for them individually, then the pair of hospitals is also acceptable for the couple, and if one member is willing to be unmatched, then any assignment of the other member to an acceptable hospital is acceptable for the couple.  The concept of \resp ness is closely related to, but a lot less restrictive than weak responsiveness as described above, even together with \subcomp ness, and has been extensively studied by economists \cite{RS90}.

Our main result is a novel and surprising polynomial-time algorithm to find a \emph{near-feasible} stable matching in an \hrc\ instance if the couples' preferences are \resp\ and \subcomp.  In this paper, our notion of \emph{near-feasibility} is based on modifying the capacities of each hospital by at most 1.  This strengthens Nguyen and Vohra's result \cite{nguyen2018near}, albeit for a special case of \hrc, in two ways: (i) capacities are varied by at most 1 rather than 2, and (ii) we provide a polynomial-time algorithm to produce the desired outcome.  This result is established via a reduction to the {\sc Stable Fixtures problem} (\fixt), a non-bipartite many-to-many generalisation of \hr\ \cite{IS07}.  Our algorithm has other nice properties: for example, it can indeed be very beneficial for a couple to apply together, as it can happen that by applying alone, only one of them would have been accepted, but by applying together, both of them are. We illustrate this property with an example. 

Next, we provide another polynomial-time algorithm for \hrc\ in the presence of \resp\ and \subcomp\ preferences that can find a stable matching if all couples are one of several possible types.  One of these types corresponds to the very practical and natural restriction of \hrcdual\ in the case of \resp\ and \subcomp\ preference lists, and gives a contrast to the NP-hardness result of Ng and Hirschberg for general \hrcdual\ instances as mentioned earlier. 
Using our approach, we argue that this algorithm can potentially be extended to other types
of couples, depending on the specific application.

We also show that our algorithm implies the polynomial-time solvability of a stable b-matching problem,  where the underlying graph is a multigraph with loops, which we call \multiSM. This is a generalisation of the \textsc{Stable Multiple Activities} problem studied by Cechlárová and Fleiner \cite{cechlarova2005generalization} with loops also allowed, where a loop occupies two positions at its incident node. On the structural side, we prove that a version of the Rural Hospitals Theorem remains true even in our \hrc\ setting with \resp\ and \subcomp\ preferences, and couples belonging to one of several possible types.  These are the first non-trivial classes of \hrc\ instances that we are aware of where these structural properties hold. 

We complement our positive results with several hardness results.
We show that \hrc\ is NP-hard, even with unit hospital capacities, short preference lists and other strong (simultaneous) restrictions, including (i) \resp\ and \subcomp\ couples, and (ii) dual markets and master preference lists \cite{IMS08}.  Hence, even in these settings we may not hope to find an exact stable matching in polynomial time, so our algorithm to find a near-feasible stable solution becomes even more appealing. 

Finally, we show that {\sc min bp hrc} is not approximable within $m^{1-\varepsilon}$, for any $\varepsilon>0$, where $m$ is the total length of the hospitals' preference lists, unless P=NP, even if each couple applies to only one pair of hospitals.  This strengthens a claim made in \cite{MMT17} that {\sc min bp hrc} is not approximable within $n_D^{1-\varepsilon}$, for any $\varepsilon>0$, where $n_D$ is the total number of doctors, unless P=NP, even if each couple applies to only one pair of hospitals.\footnote{
The proof of this result is based on a claim appearing in the supplementary material of \cite{MMT17} (Theorem 8) that \hrc\ is NP-hard even if each couple applies to only one pair of hospitals.  The latter claim is erroneous and indeed \hrc\ is solvable in polynomial time in this case as we show in Theorem \ref{thm:12abc}.}

\begin{table}
\hspace{-0.65cm}
\begin{tabular}{|c||c|c|c|}
\hline {\bf Problem} & {\bf General} & {\bf Sub-resp.\ \&} & {\bf Type-a, -b or -c} \\
               & {\bf couples} & {\bf sub-comp.\ couples} & {\bf couples} \\
 \hline
 \hline
    \hrc\ & NP-c (pref lists $\leq 2$) \cite{biro2014hospitals} & NP-c [Thm.~\ref{thm:respNPc}] & P [Thm.~\ref{thm:12abc}] \\
    \hline
    \begin{tabular}{c}\hrcdual\end{tabular} & \begin{tabular}{c} NP-c (pref lists $\leq 3$ \\ \& master lists) [Thm.\ \ref{33comdmml}]\end{tabular} & \begin{tabular}{c} P [Cor.~\ref{cor:hrcdual}]\end{tabular} &  \begin{tabular}{c} P [Cor.~\ref{cor:hrcdual}]\end{tabular} \\
    \hline
    \hrc-\textsc{near-feasible} & EXPTIME ($\pm 2$) \cite{nguyen2018near}) & P ($\pm 1$) [Thm.~\ref{thm:resp_subcom}] & P ($\pm 1$) [Thm.~\ref{thm:resp_subcom}] \\
    \hline
    \begin{tabular}{c} {\sc min bp hrc} \end{tabular} & \begin{tabular}{c} NP-h, inapprox.\ (couple \\ pref lists $\leq 1$) [Thm \ref{thm:min-bp-approx}] \end{tabular} & NP-h [Thm.~\ref{thm:respNPc}] & \begin{tabular}{c} NP-h, inapprox, \\ type a: [Thm \ref{thm:min-bp-approx}] \end{tabular} \\              
    \hline
\end{tabular}
\caption{A summary of results in this paper. \textsc{hrc-near-feasible} denotes the problem of finding a near-feasible stable matching in a \hrc\ instance, which means that the capacities of the hospitals are allowed to be changed by at most the value indicated.}
\label{tab:results}
\end{table}
The results in this paper are summarised in Table \ref{tab:results}.
The polynomial-time algorithms in this paper substantially push forward our knowledge of tractable special cases of \hrc.
Moreover they give additional evidence as to why successful matching schemes such as the NRMP have continued to find stable matchings even in the presence of couples.
Our algorithm for the case that the couples' preferences are \resp\ and \subcomp, and the couples are one of several possible types, also helps to identify the frontier between polynomial-time solvable and NP-hard cases, as our hardness results show that if we have weaker restrictions on the couples' preferences, then it becomes NP-hard to find a stable matching. 

\subsection{Structure of the paper}
The remainder of this paper is structured as follows.  Section \ref{sec:prelim} provides definitions of key notation, terminology and problem models, and motivation for sub-responsiveness, sub-completeness and
other restrictions on couples.  Our polynomial-time algorithms and structural results are then presented in Section \ref{sec:algs}, whilst our NP-hardness and inapproximability results appear in Section \ref{sec:hardness}.  Finally, Section \ref{sec:conc} gives some conclusions and open problems. 

\section{Preliminaries}
\label{sec:prelim}
\subsection{Notation, terminology and problem models}
An instance $I$ of the classical {\sc Hospitals / Residents problem} ({\sc hr}) \cite{GS62} involves two sets, 
namely a set $D=\{d_1,d_2,\dots,d_{n_D}\}$ of \emph{doctors} and a set $H=\{h_1,h_2,\dots,h_{n_H}\}$ of \emph{hospitals}. Each doctor in $D$ seeks to be assigned to a hospital, whilst there is a vector $\mathbf{q}$ that gives each hospital $h_j\in H$ a \emph{capacity} $q_j\in \mathbb{Z}^+$, 
indicating the maximum number of doctors who could be assigned to $h_j$.  Each doctor $d_i\in D$ has a strict ranking $\succ_{d_i}$ over a subset $A(d_i)$ of $H$, his/her \emph{acceptable} hospitals.  Similarly each hospital $h_j\in H$ has a strict ranking $\succ_{h_j}$ over those doctors $d_i$ such that $h_j\in A(d_i)$.

Let $M$ be a set of doctor--hospital pairs, i.e., a subset of $D\times H$.  For each $d_i\in D$ we denote $\{h_j\in H : (d_i,h_j)\in M\}$ by $M(d_i)$, whilst for each $h_j\in H$ we denote $\{d_i\in D : (d_i,h_j)\in M\}$ by $M(h_j)$.  We say that $d_i\in D$ is \emph{unmatched} in $M$ if $M(d_i)=\emptyset$.  Also we say that $h_j\in H$ is \emph{undersubscribed}, \emph{full} or \emph{oversubscribed} in $M$ according as $|M(h_j)|<q_j$, $|M(h_j)|=q_j$ and $|M(h_j)|>q_j$, respectively.
The set $M$ is a \emph{matching} if (i) $(d_i,h_j)\in M$ only if $h_j\in A(d_i)$; (ii) $|M(d_i)|\leq 1$ for each $d_i\in D$; and (iii) $|M(h_j)|\leq q_j$ for each $h_j\in H$.  In other words, a doctor is only assigned to an acceptable hospital, each doctor is assigned to at most one hospital, and no hospital exceeds its capacity.  Given a matching $M$ where $(d_i,h_j)\in M$, with a slight abuse of notation we let $M(d_i)$ denote $h_j$.

A pair $(d_i,h_j)\in D\times H$ \emph{blocks} a matching $M$, or is a \emph{blocking pair} for $M$, if
\begin{enumerate}
\item $h_j\in A(d_i)$.
\item $d_i$ is unmatched, or $h_j\succ_{d_i} M(d_i)$.
\item $h_j$ is undersubscribed, or $d_i\succ_{h_j} d_k$ for some $d_k\in M(h_j)$.
\end{enumerate}
A matching is \emph{stable} if it admits no blocking pair.

An instance $I$ of the {\sc Hospitals / Residents problem with Couples} (\hrc) involves a set $D=\{d_1,d_2,\dots,d_{n_D}\}$ of \emph{single doctors}, a set $C=\{C_i=(c_{2i-1},c_{2i}) : 1\leq i\leq n_C\}$ of \emph{couples}, and a set $H=\{h_1,h_2,\dots,h_{n_H}\}$ of \emph{hospitals}.  Let $C'$ denote the members of the couples, i.e., $C'=\{c_{2i-1},c_{2i} : 1\leq i\leq n_C\}$.  As in the \hr\ case, each hospital $h_j\in H$ has a \emph{capacity} $q_j\in \mathbb{Z}^+$.  We use the notation $r_i\in D\cup C'$ to denote either a single doctor $r_i\in D$ or a member of a couple $r_i\in C'$, and collectively the members of $D\cup C'$ are referred to as \emph{doctors}.

We define a capacity vector $\textbf{q}'$ to be \emph{near-feasible} in an instance $I$ of \hrc\ if it holds that $|q'_j-q_j|\le 1$ for all $h_j\in H$, that is, the capacity of each hospital is changed by at most 1 in $\textbf{q}'$.   

Each single doctor $d_i\in D$ has a strict ranking $\succ_{d_i}$ over a subset $A(d_i)\subseteq H$ of \emph{acceptable} hospitals, as in the \hr\ case.  Each couple $C_i=(c_{2i-1},c_{2i})$ ($1\leq i\leq n_C$) has a set $A(C_i)\subseteq ((H\cup \{\emptyset\})\times (H\cup \{\emptyset\}))\setminus \{(\emptyset,\emptyset)\}$ of \emph{acceptable} pairs of hospitals, and has a strict ranking $\succ_{C_i}$ over $A(C_i)$, where $\emptyset$ represents the case that a doctor is unmatched.  Each entry in this list is an ordered pair of the form $(h_j,\emptyset)$ $(\emptyset,h_k)$, $(h_j,h_k)$, representing either one member of the couple being matched to a hospital and the other being unmatched, or both members of the couple being matched to a hospital, respectively, where $h_j$ and $h_k$ need not be distinct hospitals.  We will assume that $\emptyset$ is a hospital with capacity $n_D+2n_C+1$ (i.e., one more than the number of doctors).  We let
\begin{eqnarray*}
A(c_{2i-1}) & = & \{h_j\in H : (h_j,\emptyset)\in A(C_i)\vee (h_j,h_k)\in A(C_i)\mbox{ for some }h_k\in H\}\backslash \{\emptyset\}\\
A(c_{2i}) & = & \{h_k\in H: (\emptyset,h_k)\in A(C_i)\vee (h_j,h_k)\in A(C_i)\mbox{ for some }h_j\in H\}\setminus \{\emptyset\}
\end{eqnarray*}
 denote the acceptable hospitals in $H$ for $c_{2i-1}$ and $c_{2i}$, respectively. 
 Finally, each hospital $h_j\in H$ has a strict ranking $\succ_{h_j}$ over those doctors $r_i\in D\cup C'$ such that $h_j\in A(r_i)$. 
 For any single doctor, couple or hospital $a$, we omit the subscript on $\succ_a$ if the identity of $a$ is clear from the context.

An \hrc\ instance can therefore be represented as a tuple $I=(D,C,H,\mathbf{q},\succ)$, where $D$ is the set of single doctors, $C$ is the set of couples, $H$ is the set of hospitals, $\mathbf{q}$ is the capacity vector for the hospitals, and $\succ$ represents the preference lists of single doctors, couples and hospitals.

In the \hrc\ setting, let $M$ be a set of doctor--hospital pairs, i.e., a subset of $(D\cup C')\times H$.  Given $r_i\in D\cup C'$, we denote $\{h_j\in H : (r_i,h_j)\in M\}$ by $M(r_i)$, and given $h_j\in H$ we denote $\{r_i\in D\cup C' : (r_i,h_j)\in M\}$ by $M(h_j)$.  We say that $r_i\in D\cup C'$ is \emph{unmatched} in $M$ if $M(r_i)=\emptyset$.  Similarly $C_i\in C$ is said to be \emph{unmatched} in $M$ if both $c_{2i-1}$ and $c_{2i}$ are unmatched in $M$.  The definitions of \emph{undersubscribed}, \emph{full} and \emph{oversubscribed} for a hospital are the same as in the \hr\ case.  Given $C_i=(c_{2i-1},c_{2i})\in C$, we define 
\[
\begin{array}{rclcl}
M(C_i) & = & 
\{(h_j,h_k) & : & (c_{2i-1},h_j)\in M\wedge (c_{2i},h_k)\in M\}~\cup\\
& & \{((h_j,\emptyset) & : & (c_{2i-1},h_j)\in M\wedge \mbox{$c_{2i}$ is unmatched in $M$}\}~\cup\\
& & \{(\emptyset,h_k) &  : & (c_{2i},h_k)\in M\wedge \mbox{$c_{2i-1}$ is unmatched in $M$}~\cup\\
& & \{(\emptyset,\emptyset) &  : & \mbox{$c_{2i-1}$ and $c_{2i}$ are unmatched in $M$}\}.
\end{array}
\]
%

A \emph{matching} is a subset $M$ of $(D\cup C')\times H$ such that (i) $M(d_i)\subseteq A(d_i)$ for each $d_i\in D$; (ii)  $M(C_i)\subseteq A(C_i)\cup \{(\emptyset,\emptyset)\}$ for each $C_i\in C$; (iii) $|M(r_i)|\leq 1$ for each $r_i\in D\cup C'$; and (iv) $|M(h_j)|\leq q_j$ for each $h_j\in H$.  Given a matching $M$ where $(r_i,h_j)\in M$ for some $r_i\in D\cup C'$, as before with a slight abuse of notation we let $M(r_i)$ denote $h_j$, and if $M(C_i)=\{p\}$ for some $C_i\in C$ and some $p\in A(C_i)\cup \{(\emptyset,\emptyset)\}$ then we also let $M(C_i)$ denote $p$.

We now define the concept of a \emph{blocking pair} in \hrc, adopting the definition from \cite{MM10}. While the definition seems complicated, the underlying idea of it is just to consider the possible cases for single doctors, couples and hospitals to form a mutually beneficial deviation from the matching.
\begin{definition}[\cite{MM10}]
Let $I$ be an instance of \hrc\ and let $M$ be a matching in $I$.  A single doctor $d_i\in D$ forms a \emph{blocking pair} of $M$ with a hospital $h_j\in H$ if:
\begin{enumerate}
\item $(d_i,h_j)$ satisfies the definition to be a blocking pair of $M$ as in the \hr\ case.
\end{enumerate}
A couple $C_i=(c_{2i-1},c_{2i})\in C$ forms a \emph{blocking pair} of $M$ with a hospital pair $(h_j,h_k)\in \protect{(H\cup \{\emptyset\})\times (H\cup \{\emptyset\})}$ if $(h_j,h_k)\in A(C_i)$, $C_i$ is unmatched or $(h_j,h_k)\succ_{C_i} M(C_i)$, and \emph{either}:
\begin{enumerate}
\setcounter{enumi}{1}
\item $h_j=M(c_{2i-1})$ or $h_k=M(c_{2i})$, and \emph{either}
\begin{itemize}
\item[(a)] $h_j=M(c_{2i-1})$, and either $h_k$ is undersubscribed or $c_{2i}\succ_{h_k} r_s$ for some $r_s\in M(h_k)\backslash \{c_{2i-1}\}$, \emph{or}
\item[(b)] $h_k=M(c_{2i})$, and either $h_j$ is undersubscribed or $c_{2i-1}\succ_{h_j} r_s$ for some $r_s\in M(h_j)\backslash \{c_{2i}\}$.
\end{itemize}
\item $h_j\neq M(c_{2i-1})$, $h_k\neq M(c_{2i})$, and \emph{either}
\begin{enumerate}
\item[(a)] $h_j\neq h_k$, and (i) $h_j$ undersubscribed or $c_{2i-1}\succ_{h_j} r_s$ for some $r_s\in M(h_j)$, and (ii) $h_k$ undersubscribed or $c_{2i}\succ_{h_k} r_t$ for some $r_t\in M(h_k)$; \emph{or}
\item[(b)] $h_j=h_k$, and $h_j$ has at least two free posts, i.e., $q_j-|M(h_j)|\geq 2$; \emph{or}
\item[(c)] $h_j=h_k$, and $h_j$ has one free post, i.e., $q_j-|M(h_j)|=1$, and $c_s\succ_{h_j} r_t$ for some $r_t\in M(h_j)$, where $s\in \{2i-1,2i\}$, \emph{or}
\item[(d)] $h_j=h_k$, $h_j$ is full, $c_{2i-1}\succ_{h_j} r_s$ for some $r_s\in M(h_j)$, and $c_{2i}\succ_{h_j} r_t$ for some $r_t\in M(h_j)\setminus \{r_s\}$.
\end{enumerate}
\end{enumerate}
$M$ is \emph{stable} if it admits no blocking pair.
\end{definition}

In \hrc, the problem is to find a stable matching or report that none exists. An important special case of \hrc\ arises when the set $H$ of hospitals can be partitioned into two sets $H_1,H_2$, such that for each couple $C_i\in C$, each acceptable pair $(h_j,h_k)\in A(C_i)$ satisfies $h_j\in H_1$ and $h_k\in H_2$, and for each single doctor $d_i$, either $A(d_i)\subseteq H_1$ or $A(d_i)\subseteq H_2$.  We refer to this restriction of \hrc\ as \hrcdual\ (corresponding to a \emph{Dual Market}).

%
%
Given that an instance $I$ of \hrc\ may not admit a stable matching, we define {\sc min bp hrc} to be the problem of finding a matching in $I$ with the minimum number of blocking pairs.  Formally, let $\mathcal M$ denote the set of matchings in $I$. Given a matching $M\in \mathcal M$, let $bp(I,M)$ denote the set of blocking pairs with respect to $M$ in $I$.  Let $bp(I)=\min \{|bp(I,M)| : M\in \mathcal M\}$.  Define {\sc min bp hrc} to be the problem of finding, given an \hrc\ instance $I$, a matching $M\in \mathcal M$ such that 
$|bp(I,M)|=bp(I)$. 

Next we define some properties regarding the preference lists of couples.

\begin{definition}
    Let $C_i=(c_{2i-1},c_{2i})$ be a couple with joint preference list $\succ_{C_i}$.
    \begin{enumerate}
      \item We say that $\succ_{C_i}$ is \emph{\resp} if there are preference orders $\succ_{c_{2i-1}},\succ_{c_{2i}}$ over $A(c_{2i-1})\cup \{\emptyset\}$ and $A(c_{2i})\cup \{\emptyset\}$ for the individual members of the couple such that, for any three hospitals $h_j$, $h_k$, $h_q$ in $H\cup \{\emptyset\}$,  
       \begin{itemize}
        \item[(i)] if $((h_j,h_k)\in A(C_i)$ and $(h_j,h_q)\in A(C_i))$ then $((h_j,h_k)\succ_{C_i}(h_j,h_q)$ if and only if $h_k\succ_{c_{2i}}h_q)$, and
        \item[(ii)] if $((h_j,h_k)\in A(C_i)$ and $(h_q,h_k)\in A(C_i))$ then $((h_j,h_k)\succ_{C_i}(h_q,h_k)$ if and only if $h_j\succ_{c_{2i-1}}h_q)$.
        \end{itemize}
         That is, there are underlying preferences for each member of the couple, such that if we send one of them to a better hospital, then the couple also improves, assuming that the new pair is acceptable.
        \item We say that $\succ_{C_i}$ is \emph{\subcomp} if 
         $A(C_i)=(A'(c_{2i-1})\times A'(c_{2i}))\setminus \{(\emptyset,\emptyset)\}$, where, for $j\in \{0,1\}$, either $A'(c_{2i-j})=A(c_{2i-j})$ or $A'(c_{2i-j})=A(c_{2i-j})\cup \{\emptyset\}$.
         That is, if both members go to an acceptable hospital for them, then the pair of hospitals is also acceptable for the couple, and if one member is willing to be unmatched, then any assignment of the other member to an acceptable hospital is acceptable for the couple.
    \end{enumerate}
    We say that $C_i$ is \resp\ or \subcomp\ if $\succ_{C_i}$ is \resp\ or \subcomp, respectively.
\label{def:coupletypes1}
\end{definition}

As we show in Theorem \ref{thm:respNPc}, \hrc\ remains NP-hard even with \resp\ and \subcomp\ preference lists.  

However, in Theorem \ref{thm:12abc}, we provide a novel reduction to \fixt\ that allows us to deal with many different kinds of couples via a polynomial-time algorithm. To this end, we describe some special types of couples. 

\begin{definition}
    We say a \resp\ and \subcomp\ couple $C_i=(c_{2i-1},c_{2i})$ is 
    \begin{itemize}
        \item \emph{separable}, if $A(C_i) = ((A(c_{2i-1}) \cup \{ \emptyset\} )\times (A(c_{2i}) \cup \{ \emptyset\})) \setminus \{ (\emptyset ,\emptyset )\}$
        (that is, it is allowed to match only one member of the couple in a feasible matching),
        \item \emph{half-separable} if $A(C_i) = A(c_{2i-1})\times (A(c_{2i}) \cup \{ \emptyset\}) $ or $A(C_i)=(A(c_{2i-1}) \cup \{ \emptyset\} )\times A(c_{2i}) $ 
        (that is, it is allowed to leave at most one member of the couple unmatched, but this can only ever be one designated member),
     \item \emph{connected}, if $A(C_i) = A(c_{2i-1})\times A(c_{2i})$ (that is, both couple members must be matched),
        \item of \emph{type-a}, if $A(c_{2i-1})\cap A(c_{2i})=\emptyset$ (that is the two members apply to distinct hospitals),
        \item of \emph{type-b}, if it is connected, $A(c_{2i-1})\cap A(c_{2i})=\{ h\}$ for some $h\in H$, and if $c_{2i-1}\succ_h c_{2i}$, then $A(c_{2i})=\{ h\}$, otherwise $A(c_{2i-1})=\{h\}$ (that is, the worst member of the couple for $h$ only applies to $h$)
        \item of \emph{type-c}, if it is connected, $A(c_{2i-1})\cap A(c_{2i})=\{ h\}$ for some $h\in H$, and for any $h'\in A(c_{2i-1})\setminus \{h\}$, $h'\succ_{c_{2i-1}} h$ and for any $h''\in A(c_{2i})\setminus \{h\}$, $h''\succ_{c_{2i}}h$, (that is, $h$ is the worst hospital for both members of the couple). 
    \end{itemize}
For convenience, to avoid collision, we consider a couple $C_i$ that is both of type-b and -c, simply as a type-b couple.
\label{def:coupletypes2}
\end{definition}

We will show that \hrc\ is solvable in polynomial time for couples of type-a, -b or -c.  Moreover for Dual Markets,  every sub-responsive and sub-complete couple is necessarily of type-a. However, every couple being of type-a is more general than \hrcdual, as it only means that for each couple $C_i$, the acceptable hospitals for each member of $C_i$ are two disjoint sets. 
To illustrate this with a simple example, imagine that there are three hospitals $h_1,h_2,h_3$ and three couples $C_1,C_2,C_3$. $C_1$ applies only to $(h_1,h_2)$, $C_2$ applies only to $(h_2,h_3)$, and $C_3$ applies only to $(h_3,h_1)$. It is easy to see that every couple is of type-a. However, the market is not dual, as for any possible partition of $\{ h_1,h_2,h_3\}$ into two parts, there will be an application where both hospitals are from the same part.


Our results will also imply a nice corollary about a stable matching problem in multigraphs, which we now define.  The setting is a multigraph $G=(V,E)$ where loops are also allowed.  Each node $v\in V$ has a strict preference list $\succ_v$ over the edges incident to $v$.  We are also given a capacity $b(v)$ for each $v\in V$.  A \emph{b-matching} is a set $M\subseteq E$ such that $|M(v)|\leq b(v)$ for each $v\in V$, where $M(v)$ is the set of edges in $M$ that are incident to $v$.  In this multigraph b-matching model, stability is defined as follows. 
\begin{definition}
    Let $M$ be a $b$-matching in a multigraph $G$ (possibly with loops). We say that an edge $e=(u,v)\notin M$ (possibly $u=v$) \emph{blocks} $M$ if one of the following conditions holds:
    \begin{itemize}
        \item $u\ne v$, and for each $x\in \{u,v\}$, either $|M(x)|<b(x)$ or there exists an edge $f\in M(x)$ such that $e\succ_x f$;
        \item $u=v$, and $b(u)-|M(u)|\geq 2$;
        \item $u=v$, and $b(u)-|M(u)|\geq 1$, and there exists an edge $f\in M(u)$ such that $e\succ_u f$;
        \item $u=v$, and there is a distinct edge $f=(u,u)\in M$ such that $e\succ_u f$;
        \item $u=v$, and there are two edges $f\in M(u)$ and $g\in M(u)$ such that $e\succ_u f$ and $e\succ_u g$.
    \end{itemize}
We say that $M$ is \emph{stable} if there is no edge blocking $M$.
\end{definition}
We now define a decision problem based on the existence of a stable b-matching in this multigraph setting.
\bigskip

\begin{tabular}{ll}
\emph{Name:} & \multiSM \\
\emph{Instance:} & 
A multigraph $G=(V,E)$ (loops are also allowed), strict preferences $\succ_v$ for \\
\textcolor{white}{\emph{Instance:}} & 
each $v\in V$ over the edges incident to $v$, and capacities $b(v)$ for each $v\in V$.\\
\emph{Question:} & Is there a stable b-matching $M$ in $G$?
\end{tabular}
\bigskip

The {\sc Stable Fixtures} problem (\fixt) \cite{IS07}, also called the \textsc{Stable b-Matching} problem, is defined similarly, but the underlying graph $G$ is a simple graph, so it is not allowed to have loops or parallel edges, and hence the notion of stability also simplifies. The generalisation of \fixt\, where parallel edges are allowed, but loops are not, has been studied by Cechlárová and Fleiner \cite{cechlarova2005generalization}, who called this problem \textsc{Stable Multiple Activities}. We note that our new problem variant does not reduce straightforwardly to \textsc{Stable Multiple Activities}. The main difficulty lies in the fact that for a loop edge $e=(u,u)$ to block, it needs two worse edges or free positions at node $u$, so if $u$ is filled with better edges than $e$ except for one worse edge, still it is the case that $e$ is not allowed to block.

A special case of \fixt\ arises when $b(v)=1$ for all $v\in V$ -- this is the classical \textsc{Stable Roommates problem} ({\sc sr}) \cite{GS62,Irv85}.  Tan \cite{tan1991necessary} gave a polynomial-time algorithm to find a stable half-integral matching (an assignment in $\{0,\frac12,1\}$ to each edge, such that each edge with weight less than 1 in the matching is \emph{dominated}, meaning that at least one endpoint is saturated with at least as good -- possibly fractional -- edges) in any {\sc sr} instance.  This algorithm was extended to the \fixt\ setting by Fleiner \cite{fleiner2010stable}.
\subsection{Motivation for couple types}
We finally provide some motivation for the couple types given in Definitions \ref{def:coupletypes1} and \ref{def:coupletypes2}.  
Firstly, in the Dual Market example given in Section \ref{sec:intro}, where each ``couple'' is a student who applies to both a university and a company, the underlying mathematical problem is \hrcdual.
It is natural to assume that preferences are \resp\ and \subcomp\ in this case, and all couples are thus of type-a as observed above. 

Types-b and -c may seem a little more artificial. On the one hand, they will demonstrate the flexibility of our technique such that it can handle couples even if they apply to the same hospital. 
On the other hand, there may be applications related to \hrc\ where these types of (\resp\ and \subcomp) structures arise, especially in scenarios where it is better for a couple (or an agent) to obtain a joint allocation where the two assigned positions are more diverse. Imagine for example that two members of a couple apply to part-time positions, which also have associated shifts (indicating which days they have to attend). Then, a couple may prefer their shifts to be on different days, to minimise the amount of daycare needed for their children. For example, if one of them applies to shifts on days 1, 2 and 3, and the other on days 3, 4 and 5, then it is natural to assume that they express a \resp\ and \subcomp\ preference list, such that (3,3) is ranked last, hence they give a type-c couple. 
Our preliminary results show that we can handle even more types of couples using our techniques, including those that may arise in other applications, but in order not to make the proofs extremely technical, we only consider the couple types defined in Definitions \ref{def:coupletypes1} and \ref{def:coupletypes2} in this paper. 

We also note the important special case of \hrc\ in which each couple applies only to one pair of hospitals.  Then as per the proof of Corollary     \ref{cor:hrclength1} below, each couple is type-a or -b.  On the other hand, \hrc\ is NP-hard even if each couple finds acceptable only two hospital pairs \cite{biro2014hospitals}.  This helps to push forward our understanding of the frontier between polynomial-time solvability and NP-hardness for restrictions of \hrc.  Indeed, couple structures that do not make the problem NP-hard may need to be very specific, as Theorem \ref{thm:respNPc} below shows that \hrc\ remains NP-hard even with \resp\ and \subcomp\ couples, such that each member of a couple applies to only two hospitals.

As a more theoretical application, couples \mbox{types-a and -b} also arise in \multiSM, as established in Corollary \ref{cor:multiSM} below.
Finally, we note that even in cases where there may be different types of couples (beyond those defined in Definitions \ref{def:coupletypes1} and \ref{def:coupletypes2}), it may speed up heuristics and IP techniques to first compute (in polynomial time) a stable solution restricted to the single doctors and the couples of types-a, -b and -c.

\section{Polynomial-time algorithms}\label{sec:algs}

In this section we describe our main positive results. We start by describing our most practical polynomial-time algorithm that is able to find a near-feasible stable matching in any \hrc\ instance $I$, where the participating couples have \resp\ and \subcomp\ preference lists.

\begin{theorem}
\label{thm:resp_subcom}
Given an \hrc\ instance where each couple's preferences are \resp\ and \subcomp, there is always a near-feasible capacity vector $\textbf{q}'$ where \mbox{$|\textbf{q}'-\textbf{q}|_{\infty}\le 1$}, such that there is a stable matching with respect to the modified capacities. Furthermore, these modified capacities and the stable matching can be found in $\mathcal{O}(m)$ time, where $m$ denotes the total length of the preference lists of the hospitals. 
\end{theorem}
\begin{proof}
We prove our statement by providing a polynomial-time reduction of the problem to \fixt. \fixt\ is known to have an algorithm with $\mathcal{O}(|E|)$ running time \cite{IS07}, where $E$ is the set of edges in the underlying graph in the \fixt\ instance.  Then, we proceed by finding a fractional stable solution and round it off in a specific way 
to obtain a near-feasible stable matching. 

Let $I=(D,C,H,\mathbf{q},\succ )$ be an instance of \hrc\ where the couples' preferences are \resp\ and \subcomp.  We transform $I$ to an \fixt\ instance $I'$, and start by describing the nodes and the capacities in $I'$. 
\begin{itemize}
    \item For each single doctor $d_i\in D$, we just create a node $d_i$ with capacity 1.
    \item For each hospital $h_j$ with capacity $q_j$, we just create a node $h_j$ with capacity $q_j$.
    \item For each couple $C_i=(c_{2i-1},c_{2i})$ we create two nodes $c_{2i-1}$ and $c_{2i}$ with capacity 1; and if $C_i$ is connected, we create connector nodes $a_{2i-1},a_{2i},b_{2i-1},b_{2i}$ with capacity 1.
\end{itemize}
We proceed by describing the edges of $I'$.

\begin{itemize}
    \item For each single doctor $d_i$ and an acceptable hospital $h_j$, we add edge $(d_i,h_j)$.
    \item For each couple $C_i=(c_{2i-1},c_{2i})$ with acceptable hospitals $A(c_{2i-1}),A(c_{2i})$:
    \begin{itemize}
        \item we add the edges $(c_{2i-1},h_j)$ for each $h_j\in A(c_{2i-1})$, and edges $(c_{2i},h_j)$ for each $h_j\in A(c_{2i})$,
        \item if $C_i$ is half-separable, then we add an edge $(c_{2i-1},c_{2i})$,
        \item if $C_i$ is connected, we add the edges $(c_{2i-1},a_{2i-1})$,$(a_{2i-1},b_{2i-1})$, $(b_{2i-1},c_{2i})$, $(c_{2i},a_{2i})$, $(a_{2i},b_{2i})$, $(b_{2i},c_{2i-1})$.
    \end{itemize}
\end{itemize}

Finally, we give the preference lists of the created nodes over their neighbours. 

\begin{itemize}
    \item 
For each single doctor $d_i$, we just keep his preference list over the hospitals.  
\item For each hospital $h_j$, it just keeps its preference list over the doctors (it is only adjacent to those nodes in $D\cup C'$ that corresponds to doctors who apply there).
\item For each member $c_k$ of a half-separable couple, if $c_k$ is the one that can be assigned alone, then he ranks the hospitals in $A(c_k)$ according to $\succ_{c_k}$, followed by his partner $c_{k'}$ last. If $c_k$ is the member who cannot be assigned alone, then $c_k$ ranks $c_{k'}$ first, and then the hospitals in $A(c_k)$.
\item Let $c_k$ be a member of a connected couple. The preferences of $c_k$ in $I'$ are then created such that he/she ranks $a_k$ first (if it exists), followed by the hospitals in $A(c_k)$ according to $\succ_{c_k}$, and $b_{k-1}$ or $b_{k+1}$ last (depending on the parity of $k$, again, if it exists).
\item For the additional $a_k,b_k$ nodes we have that $b_{2i-1}\succ_{a_{2i-1}}c_{2i-1}$, $c_{2i}\succ_{b_{2i-1}}a_{2i-1}$, $b_{2i}\succ_{a_{2i}}c_{2i}$ and $c_{2i-1}\succ_{b_{2i}}a_{2i}$.
\end{itemize}

With a modification of the algorithm of Tan \cite{tan1991necessary}, described by Fleiner \cite{fleiner2010stable}, we can find a \emph{stable half-integral matching} $M_f$ in $I'$ in time $\mathcal{O}(|E|)=\mathcal{O}(m)$. That is, for each edge $(u,v)\in E$, $M_f(u,v)\in \{0,\frac12,1\}$, for each node $u\in V$, $\sum_{(u,v)\in E} M_f(u,v)\leq q_v$ ($q_v$ is the created capacity for $v$), and for each edge $(u,v)$ with $M_f(u,v)<1$, it holds that either $u$ or $v$ is \emph{saturated} in $M_f$ (i.e., $\sum_{(w,x)\in E} M_f(w,x)=1$ for some $w\in \{u,v\}$) with partners that are at least as good as $(u,v)$.
Then, we say that $(u,v)$ is \emph{dominated} at that node.  We start with an important observation. 
\begin{obs}
\label{obs:frac-edges}
    In $M_f$, each node is incident to either 0 or two fractional edges. 
\end{obs}
\claimproofstart
     To provide a short proof of this observation, first note that each fractional edge must be dominated at one endpoint, as $M_f$ is stable. Orient each of the edges towards a node that dominates it. Then, there can be at most one incoming arc into each node, because if there would be two, then the better one of them would not be dominated there, contradiction. Also, if there is one incoming arc to a node, then there is an outgoing arc too, as for the node to be saturated, there must be another fractional adjacent edge. Hence, there can be no nodes with two outgoing arcs, because the total number of incoming arcs taken over all nodes must equal the total number of outgoing arcs taken over all nodes. 
\claimproofend

We create an set of pairs $M$ in $I$ as follows. For each doctor $r_i\in D\cup C'$ (so $r_i$ is either a single doctor or a member of a couple), let $S_i=\{h_j\in A(r_i) : M_f(r_i,h_j)>0\}$.  If $S_i\neq \emptyset$, add to $M$ the edge $(r_i,h_j)$ such that $h_j$ is the most-preferred (according to $\succ_{r_i}$) member of $S_i$.  (If $S_i=\emptyset$, no edge incident to $r_i$ is added to $M$.)  For each hospital $h_j\in H$, we possibly modify $h_j$'s capacity in $I'$ as follows: if $h_j$ was saturated in $M_f$ then we set $q_j'=|M(h_j)|$, otherwise we set $q_j'=q_j$.

\begin{claim}
    The obtained matching $M$ corresponds to a near-feasible and stable matching in $I$ with respect to the new capacities. 
\end{claim}
\specialclaimproofstart
\textbf{Feasibility.}  First of all we have to show that $M$ defines a feasible matching in $I$.  Clearly, every doctor $r_i$ is assigned to at most one hospital in $M$. If $r_i$ is a single doctor, then it is clear that $M(r_i)$ is acceptable to him.
    
    If $r_i$ is a member of a connected couple, say (by symmetry) $c_{2i-1}$, who is matched in $M$, then we first show that his partner $c_{2i}$ is also matched in $M$. Suppose for the contrary that $c_{2i}$ is not matched in $M$. Then $c_{2i}$ is not matched to any hospital node $h_j$ with positive weight in $M_f$. As the edge $(b_{2i-1},c_{2i})$ does not block $M_f$, we obtain that $c_{2i}$ must be saturated in $M_f$. Hence $M_f(c_{2i},a_{2i})+M_f(b_{2i-1},c_{2i})=1$.
    As each of $\{ a_{2i-1},a_{2i},b_{2i-1},b_{2i},c_{2i-1},c_{2i}\}$ are a strict first choice of some other agent, we get that all of them must be saturated in $M_f$. 
    Therefore, if $x=M_f(c_{2i},a_{2i})$, then $1-x=M_f(b_{2i-1},c_{2i})=M_f(a_{2i},b_{2i})$, hence $x=M_f(b_{2i},c_{2i-1})=M_f(a_{2i-1},b_{2i-1})$ and finally $1-x=M_f(c_{2i-1},a_{2i-1})$. But this implies that $M_f(c_{2i-1},a_{2i-1})+M_f(b_{2i},c_{2i-1})=1$, which contradicts the fact that $c_{2i-1}$ is matched to a hospital in $M$. 
Therefore, we obtain that for each couple $C_i$, either neither or both of them are matched in $M$. 

If $c_k$ is a member of a half-separable couple who cannot be assigned alone, then we show that $c_k$ is assigned only if his partner $c_{k'}$ is too. Indeed, if $c_{k'}$ is not matched to a hospital node, but $c_k$ is, then $c_{k'}$ is not matched with weight 1 in $M_f$ to a hospital node, and as the edge $(c_k,c_{k'})$ is the best for $c_k$, it follows that it blocks $M_f$ in $I'$, contradiction.

Also, because every couple's preferences are \subcomp, we obtain that no matter which hospitals in $A(c_{2i-1})$ and $A(c_{2i})$ are assigned $c_{2i-1}$ and $c_{2i}$ in $M$, respectively, it must be the case that $(M(c_{2i-1}),M(c_{2i}))\in A(C_i)$.
Hence each couple finds the created allocation feasible too. 

Finally, we show that no hospital is oversubscribed in $M$. Obviously, only the hospitals whose capacity we do not change to $|M(h_j)|$ can become oversubscribed. But these hospitals were undersubscribed, and it follows by Observation \ref{obs:frac-edges} that the number of doctors it obtains in $M$ is at most $|M_f(h_j)|+1\le q_j$.

\textbf{The new capacities are near-feasible.}  Next we show that the new capacities are near-feasible, that is each capacity is changed by at most 1. Suppose we changed a hospital's capacity to $|M(h_j)|$. Then, $h_j$ was saturated in $M_f$ and it had at most two incident fractional edges and all of these edges were to nodes corresponding to doctors. Hence, after the rounding, $q_j-1=|M_f(h_j)|-1\le |M(h_j)|\le |M_f(h_j)|+1=q_j+1$.

\textbf{Stability.}  Finally we show that the constructed matching $M$ is stable in $I$ with respect to the new capacities. 

Suppose that a single doctor $d_i$ blocks with a hospital $h_j$. Then, by the creation of $M$, we know that $M_f(d_i,h_j)=0$. But, the edge $(d_i,h_j)$ did not block $M_f$, so $h_j$ was filled with better, possibly fractional partners in $M_f$, each of them being nodes corresponding to doctors. By the definition of the new capacities, $h_j$ remains saturated, and any doctor at $h_j$ now must have been at $h_j$ with positive weight in $M_f$. Therefore, $h_j$ is filled with better doctors than $d_i$, a contradiction. 

Suppose that a couple $C_i\in C$ blocks $M$ in $I$ with a pair of hospitals $(h_j,h_k)\in A(C_i)$.  The following argument holds regardless of whether $h_j=h_k$ or not.
By the assumption that the couple's preferences are \resp\ and \subcomp, 
it follows that either $C_i$ is unmatched in $M$, in which case at least one of the members' nodes is not full in $M_f$ with better edges than $(c_{2i-1},h_j)$ or $(c_{2i},h_k)$ respectively (in the case of a connected couple one of $\{ c_{2i-1},c_{2i}\}$ is not matched to his favourite partner $a_{2i}$ or $a_{2i-1}$ with weight 1 in $M_f$, as otherwise $(a_{2i},b_{2i})$ blocks $M_f$); or one of its members $c_{2i-1}$, $c_{2i}$ is at a worse hospital in $M$ according to his ranking. In both cases, one of the edges $\{ (c_{2i-1},h_j), (c_{2i},h_k)\} $ is not dominated at a couple node in $M_f$.
Suppose by symmetry it is $(c_{2i},h_k)$. By the construction of $M$, it follows that $M_f(c_{2i},h_k)=0$ must hold. However, the edge $(c_{2i},h_k)$ did not block $M_f$ in $I'$, so $h_k$ was saturated in $M_f$ with better partners than $c_{2i}$, all of them being doctors. Hence, we obtain again by our construction that $h_k$ is full in $M$ with better doctors than $c_{2i}$, so $C_i$ cannot block $M$ in $I$ with $(h_j,h_k)$.

Suppose a couple $C_i$ blocks with one hospital and an application $(h_i,\emptyset )$ or $(\emptyset ,h_j)$. In this case, $C_i$ must be separable or half-separable. Suppose by symmetry that $c_{2i}$ is the one who gets matched to a hospital in the blocking coalition. Then, if $C_i$ is half-separable, then $c_{2i}$ must be the member who can be assigned alone, so he ranks $c_{2i-1}$ last. Hence, in any case, $c_{2i}$ cannot be matched to better partners in $M_f$ than $h_j$, and $M_f(c_{2i},h_j)=0$. However, $(c_{2i},h_j)$ did not block $M_f$, so $h_j$ was saturated with better doctors in $M_f$, so it still is in $M$, contradiction. 

Therefore, we obtain that the matching $M$ is stable too, proving the theorem. 
\specialclaimproofend
\end{proof}

We have shown that \resp\ and \subcomp\ preferences are enough to guarantee the existence and polynomial-time solvability of finding a near-feasible stable solution in an \hrc\ instance.

\begin{remark}
    We first remark that this algorithm extends in a straightforward way to the case where ties are allowed in the preference lists, as a stable half-integral matching can be found even if ties are allowed (by breaking the ties arbitrarily). Furthermore, a stable matching that is stable after the tie breaking is also stable with the original preference lists. Hence, the near-feasible stable matching that the algorithm finds by breaking the ties arbitrarily is sufficient. 
\end{remark}

\begin{remark}
By the result of Nguyen and Vohra \cite{nguyen2018near}, we know that there is always a near-feasible stable matching in any \hrc\ instance, if we can change the capacities by 2 instead of 1. However, their approach to finding such a near-feasible stable matching starts by constructing a fractional stable solution, which by a recent paper of Csáji \cite{csaji2022complexity} is shown to be computationally hard (PPAD-hard) in general. Therefore, even though the existence is guaranteed, there is no efficient algorithm yet that can find a near-feasible stable solution in an arbitrary \hrc\ instance.
\end{remark}

In the \hrc\ problem, especially with \resp\ and \subcomp\ preferences, it might be unclear as to why it is beneficial for a couple to apply together to pairs of hospitals. Indeed, any stable matching would be stable even if they would have applied separately. However, a very nice feature of our algorithm as given in the proof of Theorem \ref{thm:resp_subcom} is that now it may be very beneficial for a couple to apply together, as they can obtain a better match. We demonstrate this via a simple example. 

\begin{ex}
Let $h$ be a hospital with capacity 2. Suppose there is one single doctor $d$ and a couple $(c_1,c_2)$. Every doctor applies only to $h$, which has preference list $c_1\succ_h d\succ_h c_2$.  Then, if the two members of the couple would apply as separate doctors, the output matching would be to assign $c_1$ and $d$ to $h$, and leave $c_2$ unmatched. However, if $c_1$ and $c_2$ apply together as a connected couple, then in the fractional stable matching $M_f$, $d$ is assigned to $h$ with weight 1, while $c_1$ and $c_2$ are both assigned to $h$ with weight $0.5$. Hence, in the output of the algorithm, the capacity of $h$ is increased by 1, and both $c_1$ and $c_2$ are accepted. 
\end{ex}

We continue to our other main theorem of this section. 

\begin{theorem}
\label{thm:12abc}
\hrc\ is solvable in $\mathcal{O}(m)$ time, whenever each couple is \resp, \subcomp\ and is of type-a, -b or -c, where $m$ denotes the total length of the preference lists of the hospitals.
\end{theorem}
\begin{proof}
We prove our statement by creating an \fixt\ instance using a construction that extends the polynomial-time reduction given in the proof of Theorem \ref{thm:resp_subcom}. 
Let $I=(D,C,H,\mathbf{q},\succ )$ be an instance of \hrc\ satisfying the conditions of the theorem. Reindex the couples $C_i$ of type-b and c, such that $c_{2i-1}$ is the worse member of $C_i$ for the common hospital $h$. We will use this fact to make the analysis simpler. 
We proceed by describing our reduction.

We start by describing the nodes and the capacities of the constructed \fixt\ instance $I'$.
\begin{itemize}
    \item For each single doctor $d_i\in D$, we just create a node $d_i$ with capacity 1.
    \item For each hospital $h_j$ with capacity $q_j$, we create two nodes $h_j^1$ with capacity $q_j-1$ and $h_j^2$ with capacity 1.
    For simplicity, let us assume that we add $h_j^1$ even if $q_j=1$ and thus $h_j^1$ has  capacity 0. 
    \item For each couple $C_i=(c_{2i-1},c_{2i})$ we create two nodes $c_{2i-1}$ and $c_{2i}$ with capacity 1.
    \begin{itemize}
        \item If $C_i$ is separable or half-separable, then we do not add further nodes. 
        \item If $C$ is connected, we create dummy nodes $a_{2i-1},a_{2i},b_{2i-1},b_{2i}$ with capacity 1.
    \end{itemize}
\end{itemize}
We proceed by describing the edges of $I'$.

\begin{itemize}
    \item For each single doctor $d_i$ and an acceptable hospital $h_j$, we add edges $(d_i,h_j^1),(d_i,h_j^2)$.
    \item For each couple $C_i=(c_{2i-1},c_{2i})$ with acceptable hospitals $A(c_{2i-1}),A(c_{2i})$:
    \begin{itemize}
        \item For any couple  $C_i=(c_{2i-1},c_{2i})$ we add the edges $(c_{2i-1},h_j^1),(c_{2i-1},h_j^2)$ and $(c_{2i},h_k^1)$, $(c_{2i},h_k^2)$ for $h_j\in A(c_{2i-1}),h_k\in A(c_{2i})$.
        \item If $C_i$ is connected, we further add the edges $(c_{2i-1},a_{2i-1})$, $(a_{2i-1},b_{2i-1})$, $(b_{2i-1},c_{2i})$, $(c_{2i},a_{2i})$, $(a_{2i},b_{2i})$, $(b_{2i},c_{2i-1})$.
        \item If $C_i$ is of type-b 
        then we delete the edge $(c_{2i},h^2)$, 
        where $\{h\}=A(c_{2i-1})\cap A(c_{2i})$. 
        \item If $C_i$ is of type-c or $C_i$ is half-separable, then we add an edge $(c_{2i-1},c_{2i})$.
    \end{itemize}
\end{itemize}

Finally, for each constructed node $v$ in $I'$, we create a preference list $\succ_v'$ in $I'$ for $v$ over its neighbours as follows.  (Recall that if $v=c_k\in C'$, by our assumption we have a preference order $\succ_{c_k}$ over $A(c_k)$.)  Firstly, let $\succ_v'=\succ_v$, and then modify $\succ_v'$ according to the following cases. 
\begin{itemize}
    \item For each single doctor $d_i$, we just substitute each $h_j$ in $\succ_{d_i}'$ with $h_j^1\succ' h_j^2$ in $I'$. So $d_i$ has the same preferences over the hospitals and always prefers the first copy of a hospital to the second one. 
    \item For each hospital $h_j$, 
    in case we deleted an edge $(c_{i'},h_j^2)$, then $c_{i'}$ is deleted from $\succ_{h_j^2}'$.
    \item For each couple member $c_k$ $(k\in \{2i-1,2i\}$), the preferences of $c_k$ in $I'$ are created such that:
    \begin{itemize}
        \item First, each occurrence of $h_j\in A(c_k)$ is replaced by $h_j^1\succ_{c_k}' h_j^2$ in $\succ_{c_k}'$.
        \item If $c_k$ is a member of a half-separable couple, who can be matched alone, then we append his partner's node $c_{k'}$ to the end of his preference list. If he is the member who cannot be assigned alone, then we put his partner $c_{k'}$ to the top of his preference list.
        \item If $c_k$ is a member of a connected type-a or type-b couple, then we insert $a_k$ in first place in $\succ_{c_k}'$, 
        and append $b_{k-1}$ or $b_{k+1}$ (depending on the parity of $k$) to the end of $\succ_{c_k}'$.
        \item If $c_k$ is a member of a type-b couple,
        then let $h$ be the common hospital of the two members of the couple $C_i$.  We delete $h^2$ from $\succ_{c_k}'$ if $k=2i$ (recall our assumption that $c_{2i-1}$ is the worse member of the couple for $h$).
        \item If $c_k$ is a member of a type-c couple, then again let $h$ be the common hospital of the two members of the couple $C_i$.  In $\succ_{c_k}'$, $c_k$ ranks $a_k$ first, followed by the hospitals in $h_j\in A(c_k)\setminus \{ h\}$, replacing every occurrence of $h_j$ by $h_j^1\succ_{c_k}' h_j^2$ in $\succ_{c_k}'$.  This sequence is then followed in $\succ_{c_k}'$ by $h^1$, then $c_{k'}$ ($k'=4i-1-k$), then $h^2$, and finally $b_{k-1}$ or $b_{k+1}$ last (depending on the parity of $k$). 
    \end{itemize}
\item For the additional $a_k,b_k$ nodes we have that 
\begin{itemize}
    \item $b_{2i-1}\succ_{a_{2i-1}}' c_{2i-1}$,
    \item $c_{2i}\succ_{b_{2i-1}}' a_{2i-1}$, 
    \item $b_{2i}\succ_{a_{2i}}' c_{2i}$, 
    \item $c_{2i-1}\succ_{b_{2i}}' a_{2i}$.
\end{itemize}
\end{itemize}
The constructed gadgets in $I'$ corresponding to connected type-a, type-b and type-c couples are shown in Figures \ref{fig:typea}, \ref{fig:typeb} and \ref{fig:typec}, respectively.  We prove that $I'$ admits a stable matching if and only if $I$ admits a stable matching. Furthermore, we show that we can create a stable matching $M$ from a stable matching $M'$ in $I'$ in linear time. 

\begin{figure}
    \centering
    \includegraphics[scale=0.7]{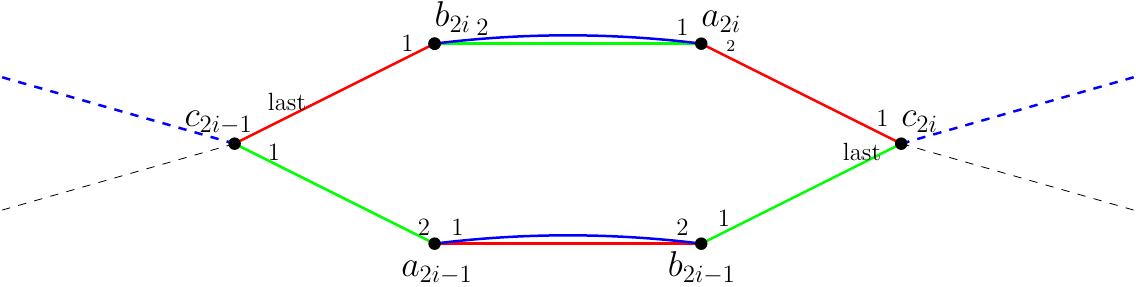}
    \caption{The gadget for a couple of type-a, with some possible stable matchings highlighted with different colours. Numbers indicate preferences.}
    \label{fig:typea}
\end{figure}
\begin{figure}
    \centering
    \includegraphics[scale=0.7]{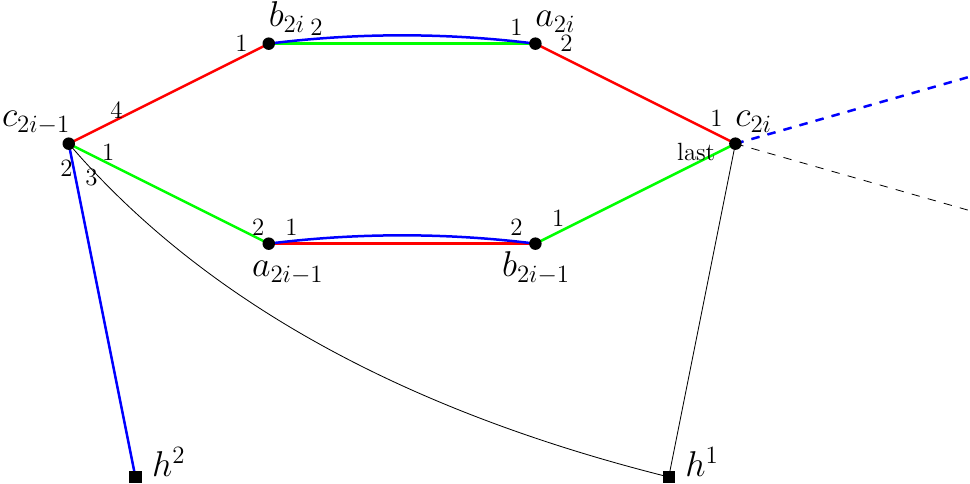}
    \caption{The gadget for a couple of type-b, with some possible stable matchings highlighted with different colours. Here, ``last" means the least-preferred partner, ``last-1" the second-least, etc.}
    \label{fig:typeb}
\end{figure}
\begin{figure}[htbp]
    \centering
    \includegraphics[scale=0.7]{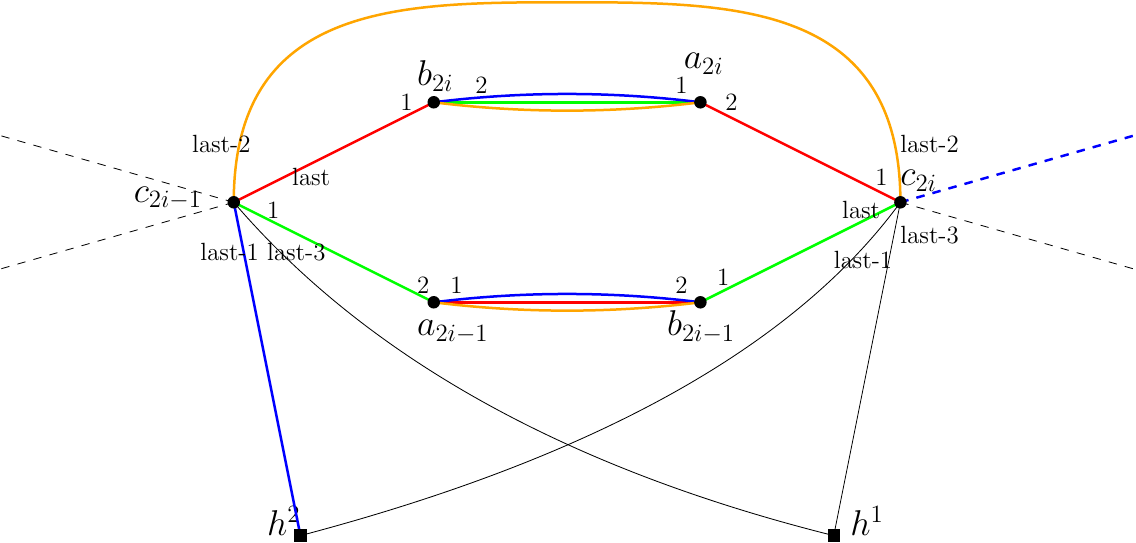}
    \caption{The gadget for a couple of type-c, with some possible stable matchings highlighted with different colours. Here, ``last" means the least-preferred partner, ``last-1" the second-least, etc.}
    \label{fig:typec}
\end{figure}

\begin{claim}
    If $I$ admits a stable matching, then so does $I'$.
    \label{claim2}
\end{claim}
\claimproofstart
    Let $M$ be a stable matching in $I$.
    
    \vspace{-8pt}
    \paragraph{Creating $M'$ from $M$.}
    First note that by adding dummy single doctors that find each hospital acceptable, and each hospital considers them worst, we can assume that each hospital is filled in $M$. (For the algorithm, we do not need to add these dummy doctors, it just helps to avoid further case distinctions in the analysis, on whether a hospital has worse doctors or free positions). 
    We create a matching $M'$ in $I'$ as follows. 

    For each hospital $h_j$, let $M^1(h_j)$ be the best $q_j-1$ doctors at $h_j$ and $M^2(h_j)$ be the worst doctor at $h_j$. Then, in $M'$, we match the nodes corresponding to the doctors in $M^1(h_j)$ to $h_j^1$ and the doctor in $M^2(h_j)$ to $h_j^2$. Note that it cannot happen that we match a doctor $c_k$, who is a member of a couple, to $h_j^2$ if the edge $(c_k,h_j^2)$ does not exist: indeed, in that case $c_k$ cannot be matched alone and his partner $c_{k'}$ only applies to $h_j$ and is worse for $h_j$, so $c_k$ cannot be the worst doctor at $h_j$. Note that as each hospital $h_j$ was filled in $M$, every node $h_j^1,h_j^2$ is filled by $M'$.

    For a single doctor, if we did not match him to a hospital, then we leave the corresponding node alone. 

    For any connected couple $C_i$, where the couple is matched in $M$, we add the edges $\{(a_{2i-1},b_{2i-1})$, $(a_{2i},b_{2i})\}$ to $M'$.

    If $C_i$ is half-separable and unmatched, then we add the edge $(c_{2i-1},c_{2i})$ to $M'$. If only one of them are matched, then we leave the unmatched node alone in $M'$.
    
    For a connected couple $C_i$ of type-a, if $C_i$ is not matched in $M$, then by the stability of $M$ and the \resp ness and \subcomp ness of the couples, it must be the case all hospital in $A(c_{2i-1})$ are filled with better doctors or all hospitals in $A(c_{2i})$ are filled with better doctors. If $C_i$ is connected, then in the former case, we add the edges $\{ (c_{2i},a_{2i}),(b_{2i},c_{2i-1}),(a_{2i-1},b_{2i-1})\}$ to $M'$, in the latter we add $\{ (c_{2i-1},a_{2i-1}),(b_{2i-1},c_{2i}),(a_{2i},b_{2i})\}$. 

    For a couple $C_i$ of type-b, if $C_i$ is not matched in $M$, then by the stability of $M$ it must be case that the common hospital $h$ (where $\{ h\} =A(c_{2i-1})$) is filled with better doctors than $c_{2i-1}$, or every hospital in $A(c_{2i})\setminus \{ h\}$ is filled with better doctors than $c_{2i}$ and $h$ has at most a single doctor that is worse than $c_{2i}$ (and also $c_{2i-1}$ as the first case does not hold). In the former case, we add the edges $\{ (c_{2i},a_{2i}),(b_{2i},c_{2i-1}),(a_{2i-1},b_{2i-1})\}$ to $M'$, in the latter we add $\{ (c_{2i-1},a_{2i-1}),(b_{2i-1},c_{2i}),(a_{2i},b_{2i})\}$.

    Finally a couple $C_i$ of type-c, if $C_i$ is not matched in $M$, then by the stability of $M$ it must be the case either all hospital in $A(c_{2i-1})$ are filled with better doctors or all hospitals in $A(c_{2i})$ are filled with better doctors or $h$ (where $\{ h\} =A(c_{2i-1})\cap A(c_{2i}))$) has a worse doctor than both members, but any other doctor at $h$ is better than both of them and every other acceptable hospital is filled with better doctors than the corresponding member of the couple. In the first case, we add the edges $\{ (c_{2i},a_{2i}),(b_{2i},c_{2i-1}),(a_{2i-1},b_{2i-1})\}$ to $M'$, in the second case we add $\{ (c_{2i-1},a_{2i-1}),(b_{2i-1},c_{2i}),(a_{2i},b_{2i})\}$, while in the third case we add $\{(c_{2i-1},c_{2i})$, $(a_{2i-1},b_{2i-1})$, $(a_{2i},b_{2i})\}$. 

\vspace{-8pt}
\paragraph{Stability of $M'$.}
    We show that $M'$ is stable.   For any couple $C_i$, the nodes $\{ a_{2i-1},a_{2i},b_{2i-1},b_{2i}\}$ cannot block, as either they have their best partner in $M'$, or they are with their other partner and their best partner is with someone better. 

    A single doctor node $d_i$ cannot block either, as he could only block with a node $h_j^1$ or $h_j^2$, such that $h_j\succ_{d_i}M(d_i)$ or $h_j=M(d_i)$ and $d_i$ is at $h_j^2$. In the first case, a blocking edge would mean that $(d_i,h_j)$ blocks $M$, in the second case it would mean that $d_j$ was not the worst doctor at $h_j$, but we matched him to $h_j^2$, both are contradictions. 

    Suppose a couple node $c_k$ blocks with someone. It cannot block with another $c_{k'}$ node, as for each connected couple of type-c, where the edge $(c_k,c_{k'})$ exists and is not in $M'$, either one of them is at their best choice (an $a_i$ agent), or both of them is at a hospital and in particular at least one of $\{ c_{k}$,$c_{k'}\}$ is at a hospital node that is better than $h^2$ (where $\{ h\} =A(c_{k})\cap A(c_{k'})$), which by the construction of the preferences is also better than his $c_k$ partner. If $c_k$ is a member of a half-separable couple, then either $(c_k,c_{k'})\in M'$ or the member who can be matched alone is matched to some hospital, which he prefers to $c_{k'}$, so it cannot block in this case either.

    A $c_k$ also cannot block with $a_i$ or $b_i$ nodes, as either $c_k$ or they are with a better partner.

In the next paragraphs we show that
   a doctor $c_k$ also cannot block with a hospital node $h_j^1$ or $h_j^2$. First we observe that if $c_k$ is with $h_j^2$, then it cannot block with $h_j^1$ as it is filled with better partners. 
   
If $c_k$ is matched to his partner's node $c_{k'}$, then we know that the couple is of type-c or is half-separable. Furthermore, in the first case, every hospital node is filled with better partners except maybe $h^2$ (where $\{h\}=A(c_k)\cap A(c_{k'}))$. If $(c_k,c_{k'})$ is of type-c, then they prefer each other to $h^2$, so the edges to $h^2$ cannot block. If they are a half-separable couple, then if $c_k$ is the member who cannot be assigned alone, then he is with his best partner, so cannot block. If he is the member who can be assigned alone, then each acceptable hospital must be filled with better doctors as $M$ was stable, so $c_k$ cannot block with a hospital node again.
    
    If $c_k$ is matched to $b_k$, by the stability of $M$ and the construction of $M'$, we know that each hospital in $A(c_k)$ was filled with better doctors in $M$, so $h_j^1,h_j^2$ is filled with better nodes than $c_k$ in $M'$ for any $h_j\in A(c_k)$; or the couple is of type-b and only $h$ (where $\{h\}=A(c_k)\cap A(c_{k'})$) has a worse doctor and at most one. Also, in this case $c_k$ is the better doctor for $h$ among the couple and so the edge $(c_k,h^2)$ does not exists. Hence, the remaining edge $(c_k,h^1)$ cannot block, as $h^1$ has only better doctors.

    If $c_k$ is matched to $a_k$, then he is already with his best partner.

If $c_k$ is not matched at all in $M'$, then he must be a member of a separable or a half-separable couple, and in the latter case the one who may not be assigned alone, while his partner $c_{k'}$ is matched to some hospital $h_j$. Hence, each adjacent hospital node must be filled with better doctors by the stability of $M$, as $c_k$ could be matched alone and his partner $c_{k'}$ cannot be at an acceptable hospital for $c_k$ (separable or half-separable couples are type-a in $I$).
    
 Otherwise $c_k$ is matched to a hospital node $h_j^1$/$h_j^2$. In this case, we claim all nodes corresponding to better hospitals are filled with better doctors, except maybe $c_k$'s partner $c_{k'}$ might be there and be a worse doctor. Otherwise, if a hospital $h_j\notin A(c_k)\cap A(c_{k'})$ is not filled with better doctors, then $c_k$ could go to $h_j$, while his partner remains at the same hospital, which would give a blocking coalition to $M$, contradiction. If there is a hospital $h$ where $\{h\}=A(c_k)\cap A(c_{k'})$ and $h$ is not filled with better doctors, and $h$ is not the worst hospital of $c_k$, then the couple is of type-b and $c_k$ is better than his partner for $h$. This also implies that we deleted the edge $(c_k,h^2)$. As $M$ was stable, $c_k$'s partner is the only worse doctor at $h$. Hence, $h^1$ is filled with better nodes than $c_k$ and the edge $(c_k,h^1)$ does not block.

    Therefore, we conclude that only hospital nodes could block with each other, but there is no edge between them, so $M'$ is stable. 
\claimproofend
    \begin{claim}
        If there is a stable matching $M'$ in $I'$, then there is a stable matching $M$ in $I$. Furthermore, $M$ can be constructed from $M'$ in linear time.
        \label{claim3}
    \end{claim}
    \claimproofstart
        Let $M'$ be a stable matching in $I'$. Again, by the worst dummy doctors, we may assume that hospital nodes are full for simplicity (having such a dummy doctor somewhere is equivalent to having that position empty in terms of blocking).
        We start by showing that for any connected couple $C_i$, it holds that $c_{2i-1}$ matched to a hospital node if and only if $c_{2i}$ is too. Suppose $c_{2i-1}$ is matched to some hospital node $h_j^1$/$h_j^2$, but $c_{2i}$ is not. Then, as $(c_{2i-1},a_{2i-1})$ does not block, $(a_{2i-1},b_{2i-1})\in M'$. As $(b_{2i-1},c_{2i})$ does not block, $(c_{2i},a_{2i})\in M'$. But then, $b_{2i}$ is unmatched and $(a_{2i},b_{2i})$ blocks $M'$, contradiction. The other case is similar. 

        If $c_k$ is a member of a half-separable couple who cannot be matched alone, then if $c_k$ would be matched, but not his partner $c_{k'}$, then the edge $(c_k,c_{k'})$ would block, as $c_k$ considers $c_{k'}$ as best. So only the designated members of half-separable couples may be matched alone.

        It is clear that in $M'$, for each hospital node $h_j^1/h_j^2$, every doctor at $h_j^1$ is better for $h_j$ than the one 
        at $h_j^2$. Otherwise that doctor matched to $h_j^2$ would block with $h_j^1$, as he prefers $h_j^1$ to $h_j^2$.
        
        We create $M$ as follows. For each doctor (either single or member of a couple) we assign him to the hospital whose node he is matched in $M'$, if he is matched to such a node, otherwise we leave him unassigned. 

        It is straightforward to see that $M$ is feasible for the hospitals and the single doctors. For the couples, we have already shown that no member remains unmatched, who is not allowed to; and as their preferences are \subcomp, no matter how their two corresponding nodes are matched to hospital nodes, it will correspond to an acceptable assignment of the couple. 

        Hence we only have to show that $M$ is stable. 

    Suppose a single doctor $d_i$ blocks $M$ with a hospital $h_j$. But then the edge $(d_i,h_j^2)$ must block $M'$, contradiction. 

    Suppose a couple $C_i$ blocks $M$ with two distinct hospitals $h_j\ne h_l$. 
    
    Then, if the couple (or at least one member) is matched in $M$ to some hospitals already, then it holds that in $I'$, either $c_{2i-1}$ or $c_{2i}$ prefer the hospitals $h_j^2$ or $h_l^2$ to their assignments, whenever the edges to $h_j^2$ and $h_l^2$ exist. If neither of $h_j,h_l$ are a common hospital for the couples, then both edges exists. Hence, $(c_{2i-1},h_j^2)$ or $(c_{2i},h_l^2)$ blocks $M'$, contradiction. If $h_j= h$ (where $\{h\}=A(c_{2i-1})\cap A(c_{2i})$), then the edges $(c_{2i-1},h_j^2),(c_{2i},h_l^2)$ exist and the same argument works. If, $h_l=h$, then the couple is connected, $c_{2i-1}$ only applies to $h$ and is matched there in $M$, so the blocking application would be $(C_i,h,h)$, contradicting our assumption on $h_j\ne h_l$. 
    
    Assume now that the couple $C_i$ is unassigned in $M$. Then, if neither of them are matched anywhere in $M'$, them the couple is separable and both $(c_{2i-1},h_j^2)$, $(c_{2i},h_l^2)$ block $M'$, contradiction. If the edge $(c_{2i},c_{2i-1})$ is in $M'$, then as $h_j\ne h_l$, one of $h_j^2,h_l^2$ is better for the corresponding member of the couple, so the edge $(c_{2i-1},h_l^2)$ or $(c_{2i},h_l^2)$ (both exist) blocks. Otherwise, one of the couple nodes must receive the worst partner $b_{2i}$ or $b_{2i-1}$, respectively (if both of them would be matched to their favourite partner $a_{2i-1},a_{2i}$, then $b_{2i}$ would be unmatched, so $(a_{2i},b_{2i})$ would block), or remain alone. This couple node, say $c_{2i}$ prefers $h_l^2$, if it exists, in which case $(c_{2i},h_l^2)$ blocks $M'$ contradiction. If the edge to $h_l^2$ does not exists, then $c_{2i}$ is the better member of the couple for their common hospital $h=h_l$ and $C_i$ is of type-b. But then, $c_{2i-1}$ only applies to $h$, so $h_j=h$, contradicting our assumption on $h_j\ne h_l$. 

    Suppose a separable or half-separable couple blocks with a single hospital (so an application $(C_i,\emptyset ,h_l)$ or $(C_i,h_j,\emptyset )$ ). Then, the member who is matched in the blocking coalition either has only edges to hospitals or $(c_{2i-1},c_{2i})$, which is his worst. Hence, the hospital node $h_l^2$ or $h_j^2$ is better for the corresponding node than his partner in $M'$, if there is any and as the application blocks $M$, $(c_{2i-1},h_j)$ or $(c_{2i},h_l)$ blocks $M'$, contradiction.

    Finally, suppose that a couple $C_i$ blocks with a double application to their common hospital $h$. 
    Then, $C_i$ is a couple of type-b or -c.
    
    If it of type-c, then as $h$ is the worst for both members, $C_i$ must be unassigned. Also, as $C_i$ blocks $M$ with $h$, there is at least one worse doctor than $c_{2i-1}$ and a distinct doctor that is worse than $c_{2i}$. So there are at least two doctors that are worse than $c_{2i}$ at $h$, so at least one of them matched to $h^1$ in $M'$. As $M'$ is stable and $(c_{2i-1},h^2)$ does not block, we get that $c_{2i-1}$ is either matched to $c_{2i}$ or to $a_{2i-1} $, in the latter case $c_{2i}$ is matched to $b_{2i-1}$. In both cases, we obtain that $(c_{2i},h^1)$ blocks. 

    If $C_i$ is of type-b, then if $C_i$ is unassigned in $M$, then by similar reasoning we obtain that there must be a doctor worse than $c_{2i}$ matched to $h^1$ and a doctor worse than $c_{2i-1}$ matched to $h^2$ in $M'$. For $(c_{2i-1},h^2)$ not to block, we get that $c_{2i}$ must be matched to $b_{2i-1}$, so $(c_{2i},h^1)$ blocks $M'$, contradiction. Finally, if $C_i$ was assigned, then we know that $c_{2i-1}$ must be at $h$ and so $c_{2i}$ is not at $h$, but prefers $h$ to $M'(c_{2i})$. As $C_i$ blocks $M$ with $(h,h)$, we get that there is a distinct doctor from $c_{2i-1}$ at $h$, who is worse than $c_{2i}$. In particular, there are at least two worse doctors at $h$ than $c_{2i}$, so the edge $(c_{2i},h^1)$ blocks $M'$, contradiction again.

    Therefore we conclude that $M$ is stable. It is also clear that $M$ can be constructed in $M'$ in time linear in the number of edges.
    \claimproofend
    
The theorem then follows from Claims \ref{claim2} and \ref{claim3} and the fact that in an \fixt\ instance, a stable solution can be found in $\mathcal{O}(|E|)$ time, if one exists \cite{IS07}.
\end{proof}
We next make a simple observation about separable couples.
If every couple were separable, then by treating each member as just a single doctor and running the doctor-oriented Gale-Shapley algorithm \cite{GI89} we could always obtain a stable matching. To see this, just observe that no matter how the two members (or just one) get matched to acceptable hospitals, by separability, it will be acceptable for the couples. Furthermore, if a couple $C$ were to block the matching with a pair of hospitals $(h,h')$, or with a pair $(h,\emptyset)$ or $(\emptyset,h)$, then at least one member of the couple has to consider the corresponding hospital better than his/her current assignment according to his own preference list. Hence, he applied to that hospital, but got rejected, so the hospital is still filled with better doctors, a contradiction. This simple observation was also proven in \cite{klaus2005stable}, although only in a framework that does not allow both members of a couple to apply to the same hospital.
This demonstrates that separable couples are easier to handle and indeed it is one of the main reasons that we consider mostly connected couple types (type-b and type-c). We could extend our technique to other separable and half-separable couple types, but it would make both the statement and the proof of Theorem \ref{thm:12abc} even more technical.

We next discuss some interesting corollaries of Theorem \ref{thm:12abc}, which indicate that special cases of \hrc\ and of \hrcdual, and \multiSM, are solvable in polynomial time.

\begin{corollary}
\hrc\ is solvable in polynomial time if each
couple only applies to one pair of hospitals.
\label{cor:hrclength1}
\end{corollary}
\begin{proof}
Let $C_i\in C$ be an arbitrary couple in $C$.  By our assumption, either (i) $A(C_i)=(h_j,\emptyset)$, (ii) $A(C_i)=(\emptyset,h_k)$, or (iii) $A(C_i)=(h_j,h_k)$, for some $h_j,h_k\in H$.  In cases (i) and (ii), clearly $C_i$ is of type-a.  In case (iii), $C_i$ is of type-a if $h_j\neq h_k$ and is of type-b if $h_j=h_k$.  The result then follows by Theorem \ref{thm:12abc}.
\end{proof}

\begin{corollary}
    The \hrcdual\ problem is solvable in polynomial time if each couple is \resp\ and \subcomp. Furthermore, there always exists a stable matching.
\label{cor:hrcdual}
\end{corollary}
\begin{proof}
Solvability in polynomial time follows immediately from Theorem \ref{thm:12abc} as in a Dual Market, each sub-responsive and sub-complete couple is of type-a.

Denote the two classes of hospitals in the Dual Market by $H_1$ and $H_2$. For each $i\in \{1,2\}$, denote the set of doctors (both single and a member of a couple), who only apply to hospitals in $H_i$ by $R_i$.

Furthermore, in relation to the \fixt\ instance $I'$ constructed in the proof of Theorem \ref{thm:resp_subcom}, in this case $A=H_1\cup R_2$ along with the neighbors $a_i,b_i$ of the couple nodes of $R_1$ in their gadget and $B=H_2\cup R_1$  along with the neighbors $a_i,b_i$ of the couple nodes of $R_2$ in their gadget give a bipartition of the nodes, such that neither of them spans any edges, so the constructed instance $I'$ is bipartite. Hence, a stable matching always exists and can be found by a Deferred Acceptance algorithm. 
\end{proof}

\begin{corollary}
\multiSM\ is solvable in polynomial time.
\label{cor:multiSM}
\end{corollary}
\begin{proof}
    We can consider such an instance $I$ as a \hrc\ instance $I'$, where each node $v$ with capacity $b(v)$ corresponds to a hospital with capacity $b(v)$ and each edge $e=(u,v)$ corresponds to distinct connected couples who have a single application $(u,v)$. The preferences of the hospitals are then created according to the corresponding node's preference lists over the adjacent edges. If there were any loops adjacent to $v$, then we break the tie between the two endpoints of a loop edge $e=(v,v)$ in an arbitrary way (it was one edge, so it was a single entry in $v$'s preference list), and the couple member corresponding to the better end of $e$ will be the better doctor for $v$. 
It is straightforward to verify that the stable matchings of this \hrc\ instance are exactly the stable matchings in $I$.

    Furthermore, in this created instance, every couple is connected and is either of type-a or of type-b, so by Theorem \ref{thm:12abc}, we deduce that the problem can be solved in polynomial time. 
\end{proof}

We can also show that a variant of the Rural Hospital theorem extends to this framework. 
\begin{corollary}
    In an \hrc\ instance where each couple is of type-a, -b or -c, it holds that in every stable matching, \\
    (i) the same set of single doctors are matched, \\
    (ii) every hospital is assigned the same number of doctors, and \\(iii) if a hospital is undersubscribed in one stable matching, then it is assigned the same set of doctors in all stable matchings. 
\end{corollary}
\begin{proof}
First, we note that by \cite[Theorem 2.1]{IS07} for \fixt, the \fixt\ instance $I'$ created in the proof of Theorem \ref{thm:12abc} satisfies the property that each node has the same number of incident edges assigned to it in any stable matching. As the nodes corresponding to single doctors are only adjacent to hospital nodes, and the nodes of the hospitals are only adjacent to doctor nodes, parts (i) and (ii) of the theorem immediately follow.

For part (iii), let $M_1,M_2$ be two stable matchings in $I$, and let $h$ be a hospital that is undersubscribed in $M_1$ (and in $M_2$ by part (ii)). Let $M_1',M_2'$ be the corresponding stable matchings in $I'$, respectively. In both of them, $h^2$ is unmatched, as $h$ is undersubscribed in $M_1,M_2$ in $I$. By \cite[Theorem 2.1]{IS07}, each of $M_1',M_2'$ contains the same number of incident edges at each node, so their symmetric difference $M_1'\oplus M_2'$ is an Eulerian graph. Hence, it is a union of edge-disjoint cycles.

Suppose for a contradiction that $h$ is assigned different sets of doctors in $M_1$ and $M_2$.  Then there is a cycle $C$ in $M_1'\oplus M_2'$ that is incident to $h^1$. As none of the edges of $C$ block $M_1'$ or $M_2'$ in $I'$, it follows that the preferences in $C$ must be cyclic, that is each node of $C$ prefers the next node to the previous one in one of the orientations of the cycle. This implies that there is a node $v$ of $C$ that prefers $h^1$ to its partner in $M_1'$ or $M_2'$. 
Assume by symmetry that $v$ prefers $h^1$ to its partner in $M_1'$. This node $v$ corresponds to a doctor $r$ (either single or a member of a couple in $I$).

If $r$ is single, clearly $(r,h)$ blocks $M_1$ in $I$, since $r$ is unmatched in $M_1$ or prefers $h$ to $M_1(r)$ in $I$, a contradiction.  Hence suppose that $r=c_k$ is a member of a couple $(c_k,c_{k'})$, without loss of generality.  If $c_{k'}$ is assigned in $M_1$, say to $h'$ (possibly $h=h'$), then, since $c_k$ prefers $h^1$ to $M_1'(c_k$), by the \resp ness and \subcomp ness of the preferences in $I$, it follows that $(c_k,c_{k'})$ blocks $M_1$ with $(h,h')$ in $I$, a contradiction.

Now suppose that the couple $(c_k,c_{k'})$ is unmatched in $M_1$. Then, each member of the couple is matched within their gadget in $M_1'$, and the cycle $C$ must pass through the couple's gadget in $I'$ in such a way that $c_k$ has their worst choice and $c_{k'}$ has their best choice.  As the preferences are cyclic in $C$, it follows 
that there is a hospital $h'$ that has a worse doctor than $c_{k'}$ in $M_1'$. If $h\ne h'$, it follows that $(c_k,c_{k'})$ blocks $M_1$ with $(h,h')$ in $I$, a contradiction. If $h=h'$, then $h$ has a free place and a worse doctor than $c_{k'}$, so $(c_k,c_{k'})$ blocks $M_1$ with $(h,h)$ in $I$ again, a contradiction. 
\end{proof}

\begin{remark}
    It does not hold that the same set of doctors are matched in every stable matching, even if each couple is of type-a, -b or -c. Take an instance with a single hospital $h$ with preference $c_1\succ c_3\succ c_4 \succ c_2$ and capacity 2, and connected couples $C_1=(c_1,c_2)$, $C_2=(c_3,c_4)$, each of which finds acceptable only the pair $(h,h)$. Then, no matter which couple gets the two positions at $h$, we obtain a stable matching.  Note that both couples are of type-b.
\end{remark}

\section{Hardness Results}\label{sec:hardness}
In this section we provide hardness results for several variants of \hrc.  In Section \ref{sec:respsubhard}, we show that \hrc\ is NP-hard even for \resp\ and \subcomp\ preference lists, for unit hospital capacities, and even in the presence of additional restrictions.  In Section \ref{sec:dualhard}, we show that \hrcdual\ is NP-hard even for unit hospital capacities, and even if every preference list is of length 3 and the preferences of single doctors, couples and hospitals are derived from \emph{master lists} \cite{IMS08} (that is, uniform strict rankings of all hospitals, all acceptable pairs of hospitals, and individual doctors, respectively).  Finally, in Section \ref{sec:inapprox}, we show that {\sc min bp hrc}, the problem of finding a matching with the minimum number of blocking pairs, given an instance of \hrc, is very hard to approximate even when each hospital has unit capacity and each couple finds only one pair of hospitals acceptable.
\subsection{Sub-responsive and \subcomp\ preferences}
\label{sec:respsubhard}
It might be tempting to believe that \resp\ and \subcomp\ preferences are sufficient to guarantee the polynomial-time solvability of \hrc, since they allow us to separate each couple's joint preference list over pairs over hospitals into two different preference lists over single hospitals. However, it turns out that these assumptions are unfortunately not enough.

In this section we show that \hrc\ remains NP-hard even with \resp\ and \subcomp\ couple preferences.  We provide a reduction from {\sc com smti}, which is the problem of deciding, given an instance of {\sc Stable Marriage with Ties and Incomplete lists}, whether a complete stable matching (i.e., a stable matching that matches all agents) exists (this problem is defined formally in \cite[Section 1.3.5]{Man13}, for example).
\begin{theorem}
\label{thm:respNPc}
\hrc\ is NP-hard even if each couple has \resp\ and \subcomp\ preferences. This holds even with unit hospital capacities, and even if each preference list is of length at most 4.  
\end{theorem}
\begin{proof}
 We reduce from {\sc com smti}, which is shown to be NP-complete in the Appendix of \cite{MM10} even if every man's list is strictly ordered and of length at most 3, and every woman's list is either strictly ordered and of length at most 3, or is a tie of length 2.  Hence let $I$ be an instance of this restriction of {\sc com smti}.

Let $U$ be the set of men in $I$ and let $W$ be the set of women in $I$.  Let $G=(U\cup W,E)$ denote the underlying bipartite graph of $I$, and let $m_I$ denote $|E|$.  Let $W^s$ denote the set of women in $W$, whose preference lists are strictly ordered in $I$, and let $W^t=W\setminus W^s$.  Then, every woman in $W^t$ has a preference list in $I$ that is a tie of length 2. 

We create an instance $I'$ of \hrc\ as follows. For each woman $w\in W^s$, we create a hospital $h_w$  with capacity 1. For each woman $w\in W^t$, we create three hospitals $h_w^1,h_w^2,h_w^3$ each with capacity 1, and a couple $(c_w^1,c_w^2)$. Finally, for each man $u\in U$, we create a single doctor $d_u$, dummy hospitals $h_u^1,h_u^2,h_u^3$ each with capacity 1, a dummy hospital $f_u$ with capacity 1, and three dummy couples $(c_u^1,c_u^2)$, $(c_u^3,c_u^4),(c_u^5,c_u^6)$.

The preferences lists for the agents created are shown in Figure \ref{fig:hrc-resp-subc}.
\begin{figure}[t]
\[
\begin{array}{rll}
d_u : & P(u)\succ f_u & (u\in U) \\
(c_u^1,c_u^2) : & (f_u,h_u^2)\succ (h_u^1,h_u^2) & (u\in U) \\
(c_u^3,c_u^4) : & (h_u^2,h_u^3) & (u\in U) \\
(c_u^5,c_u^6) : & (h_u^3,h_u^1) & (u\in U) \\
(c_w^1,c_w^2): & (h_w^2,h_w^2)\succ (h_w^1,h_w^2)\succ (h_w^2,h_w^3)\succ (h_w^1,h_w^3)  & (w\in W) \\
\\
h_w : & P(w) & (w\in W^s) \\
h_w^1 : & c_w^1\succ d_u & (w\in W^t) \\
h_w^2 : & c_w^1\succ c_w^2 & (w\in W^t)\\
h_w^3: & c_w^2\succ d_{u'} & (w\in W^t) \\
h_u^1 : & c_u^6\succ c_u^1 & (u\in U)\\
h_u^2 : & c_u^2\succ c_u^3 & (u\in U)\\
h_u^3 : & c_u^4\succ c_u^5 & (u\in U)\\
f_u: & d_u\succ c_u^1 & (u\in U) 
\end{array}
\]
\caption{Preference lists in the constructed instance of \hrc\ with \resp\ and \subcomp\ preferences.}
\label{fig:hrc-resp-subc}
\end{figure}
Here, for a man $u\in U$, $P(u)$ denotes $u$'s preference list in $I$, with every woman $w$ replaced by a hospital corresponding to $w$ as follows: if $w\in W^s$, then we replace $w$ by $h_w$, whilst if $w\in W^t$ and $u$ precedes $w$'s other neighbour $u'$ in some fixed ordering of $U$ then we replace $w$ by $h_w^1$, otherwise we replace $w$ by $h_w^3$.
Similarly, for a woman $w\in W^s$, $P(w)$ denotes $w$'s preference list in $I$, with every man $u$ replaced by doctor $d_u$.  For a woman $w\in W^t$ with neighbors $u$ and $u'$, with $u$ coming before $u'$ in the fixed ordering of $U$, $h_w^1$ only ranks $d_u$ and $h_w^3$ only ranks $d_{u'}$.  Since $P(u)$ and $P(w)$ both have length at most 3 given the restricted version of {\sc com smti} that we are reducing from, it is straightforward to verify that the preference list of each single doctor and couple is of length at most 4, whilst the preference list of each hospital is of length at most 3.

This concludes the construction of $I'$. It is straightforward to verify that each couple has a \resp\ and \subcomp\ preference list.

\begin{claim}
    If $I$ has a complete stable matching, then $I'$ has a stable matching.
\end{claim}
\claimproofstart
    Let $M$ be a complete stable matching in $I$. We create a matching $M'$ in $I'$ as follows. For each $(u,w)\in M$ we add the pair $(d_u,h_w)$ to $M'$, if $w\in W^s$, and the pair $(d_u,h_w^1)$ or $(d_u,h_w^3)$ to $M'$, if $w\in W^t$ (depending on which one of these two hospitals is acceptable to $d_u$). Then, for each $w\in W^t$, if $h_w^1$ has a single doctor $d_u$, we match the couple $(c_w^1,c_w^2)$ to $(h_w^2,h_w^3)$, otherwise we match the couple to $(h_w^1,h_w^2)$. Finally, for each $u\in U$, we match $(c_u^1,c_u^2)$ to $(f_u,h_u^2)$, and $(c_u^5,c_u^6)$ to $(h_u^3,h_u^1)$. 

    We claim that $M'$ is stable in $I'$. Suppose first that a single doctor $d_u$ blocks $M'$ in $I'$ with a hospital $h_w$ or $h_w^j$, $j\in \{ 1,3\}$. It cannot be $f_u$, as $M$ was complete, so $d_u$ is at a better hospital. Hence the blocking hospital corresponds to a woman that $u$ prefers to his partner in $M$. Also, for $h_w^j$ with $j\in \{ 1,3\}$, either this hospital is assigned the single doctor $d_u$ or it is assigned a member of the couple $(c_w^1,c_w^2)$, whom it prefers to $d_u$. Hence, such a hospital cannot block with $d_u$. Therefore the blocking hospital must be $h_w$ with $w\in W^s$. But then, we obtain that the edge $(u,w)$ must have blocked $M$ in $I$, a contradiction. 

    Suppose now that a couple $(c_w^1,c_w^2)$ blocks $M'$ in $I'$. As $h_w^2$ has capacity 1, blocking can only happen if we assigned them to $(h_w^2,h_w^3)$ and they block with $(h_w^1,h_w^2)$. But then, $h_w^2$ obtains its best doctor in $M'$, so this coalition cannot block $M'$ in $I'$ after all, a contradiction. 

    Finally, a couple $(c_u^3,c_u^4)$ also cannot block $M'$ in $I'$, because $h_u^2$ has a better doctor. The other couples $(c_u^1,c_u^2)$ and $(c_u^5,c_u^6)$ obtain their best option in $M'$.

    Hence, no doctor or couple can block $M'$ in $I'$, so $M'$ is stable in $I'$. 
\claimproofend

\begin{claim}
    If $I'$ has a stable matching, then $I$ has a complete stable matching. 
\end{claim}
\specialclaimproofstart
    Let $M'$ be a stable matching in $I'$. First we show that each doctor $d_u$ is matched to some $h_w$ or $h_w^j$ hospital. Suppose this is not the case for some doctor $d_u$.  Then, as $d_u$ is the first choice of $f_u$, it follows that $(d_u,f_u)\in M'$. But then, there is no stable allocation of $(c_u^1,c_u^2), (c_u^3,c_u^4),(c_u^5,c_u^6)$ to $h_u^1,h_u^2,h_u^3$. 
    Indeed, it is easy to verify that each of the three possible matchings is blocked by a couple.

    Next, we show that there is no woman $w\in W$ such that the hospital(s) corresponding to $w$ receive more then one doctor $d_u$ in $M'$. For a woman in $W^s$, this is trivial. Let $w\in W^t$ and suppose to the contrary that both $h_w^1$ and $h_w^3$ are assigned a single doctor in $M'$. 
    Then, $(c_w^1,c_w^2)$ cannot be matched, so the couple blocks $M'$ in $I'$ with $(h_w^2,h_w^3)$, contradiction. 

    Hence $M=\{(u,w)\in U\times W : (d_u,h_w)\in M'\vee (d_u,h_w^1)\in M'\vee (d_u,h_w^3)\in M'\}$ is a complete matching in $I$.  It remains to show that $M$ is stable in $I$. Suppose that $(u,w)$ blocks $M$ in $I$. Then, $w\in W^s$ and the pair $(d_u,h_w)$ blocks $M'$ in $I'$, a contradiction again. Hence $M$ is a complete stable matching in $I$.
\specialclaimproofend


\end{proof}


\subsection{Dual markets}
\label{sec:dualhard}
In this section we present our hardness results for \hrcdual, even in the case when both sides have master lists.  We begin by showing that, given an instance of \hrcdual, it is NP-complete to decide whether a \emph{complete} stable matching (i.e., a stable matching in which all single doctors and couples are matched) exists.

\begin{theorem}
\label{33-COM-HRC}
Given an instance of \hrcdual, the problem of deciding whether there exists a complete stable matching is NP-complete. The result holds even if all preference lists have length 3 and all hospitals have capacity 1.
\end{theorem}
\begin{proof}
The problem is clearly in NP, as a given assignment may be verified to be a complete, stable matching in polynomial time. 

We use a reduction from a restricted version of {\sc sat}.  More specifically, let \sat\ denote the problem of deciding, given a Boolean formula $B$ in CNF over a set of variables $V$, whether $B$ is satisfiable, where $B$ has the following properties: (i) each clause contains exactly 3 literals and (ii) for each $v_i\in V$, each of literals $v_i$ and $\bar{v_i}$ appears exactly twice in $B$. Berman et al.\ \cite{BKS03} showed that \sat\ is NP-complete. 

Let $B$ be an instance of \sat, where $V=\{v_0,v_1,\dots,v_{n-1}\}$ is the set of variables in $B$ and $C_B=\{c_1,c_2,\dots,c_m\}$ is the set of clauses in $B$.  Then for each $v_i\in V$, each of literals $v_i$ and $\bar{v_i}$ appears exactly twice in $B$. Also $|c_j|=3$ for each $c_j\in C_B$. (Hence $m=\frac{4n}{3}$.)
 
We construct an instance $I$ of \hrcdual\ using a similar reduction to that employed by Irving et al.\ \cite{IMO09} to show that {\sc com smti} is NP-complete even if each preference list is of length at most 3.

The set of doctors in $I$ is $X\cup K\cup P\cup S\cup T$, where $X=\cup_{i=0}^{n-1} X_i$, $X_i=\{x_{4i+r} : 0\leq r\leq 3\}$ ($0\leq i\leq n-1$), $K=\cup_{i=0}^{n-1} K_i$, $K_i=\{k_{4i+r} : 0\leq r\leq 3\}$ ($0\leq i\leq n-1$), $P=\cup_{j=1}^m P_j$, $P_j=\{p_j^r : 1\leq r\leq 6\}$ ($1\leq j\leq m$), $S=\{s_j : c_j\in C_B\}$ and $T=\{t_j : c_j\in C_B\}$. The doctors in $S\cup T$ are single and the doctors in $X\cup K\cup P$ are involved in couples.

The set of hospitals in $I$ is $Y\cup L\cup C_B'\cup Z$, where $Y=\cup_{i=0}^{n-1} Y_i$, $Y_i=\{y_{4i+r} : 0\leq r\leq 3\}$ ($0\leq i\leq n-1$), $L=\cup_{i=0}^{n-1} L_i$, $L_i=\{l_{4i+r} : 0\leq r\leq 3\}$ ($0\leq i\leq n-1$), $C_B'=\{c_j^s : c_j\in C_B\wedge 1\leq s\leq 3\}$ and $Z=\{z_j^r : 1\leq j\leq m\wedge 1\leq r\leq 5\}$.   All hospitals in $I$ have capacity 1.

In the joint preference list of a couple $(x_{4i+r}, k_{4i+r}) (x_{4i+r}\in X, k_{4i+r}\in K)$, if $r\in \{0,1\}$, the symbol $c(x_{4i+r})$ 
denotes the hospital $c_j^s\in C_B'$ such that the $(r+1)$th occurrence of literal $v_i$ appears at position $s$ of $c_j$.  Similarly, if $r\in \{2,3\}$ then the symbol $c(x_{4i+r})$ denotes the hospital $c_j^s\in C_B'$ such that the $(r-1)$th occurrence of literal $\bar{v_i}$ appears at position $s$ of $c_j$. Also in the preference list of a hospital $c_j^s\in C_B'$, if literal $v_i$ appears at position $s$ of clause $c_j\in C_B$, the symbol $x(c_j^s)$ denotes the doctor $x_{4i+r-1}$ where $r=1,2$ according as this is the first or second occurrence of literal $v_i$ in $B$.  Otherwise
if literal $\bar{v_i}$ appears at position $s$ of clause $c_j\in C_B$, the symbol $x(c_j^s)$ denotes the doctor $x_{4i+r+1}$ where $r=1,2$ according as this is the first or second occurrence of literal $\bar{v_i}$ in $B$.  

The preference lists of the doctors and hospitals in $I$ are shown in Figure \ref{preflists33}.  Clearly each preference list is of length at most 3.  The doctors in $P_j\cup \{c_j^s : 1\leq s\leq 3\}\cup \{z_j^s : 1\leq s\leq 5\}\cup \{s_j,t_j\}$ constitute a \emph{clause gadget} for each $j$ ($1\leq j\leq m$) and such a gadget is illustrate in Figure 
\ref{image:clausegadget}.  Similarly the doctors in $X_i\cup K_i$ and the hospitals in $Y_i\cup L_i$ constitute a \emph{variable gadget} for each $i$ ($0\leq i\leq n-1)$ and such a gadget is illustrated in Figure \ref{image:variablegadget}.  In each of Figures \ref{image:clausegadget} and \ref{image:variablegadget}, two doctors involved in a couple are shown using dashed lines.  Numbers beside full lines indicate preference list rankings.

\begin{figure}
\[
\begin{array}{rll}
(x_{4i}, k_{4i}) : & (y_{4i}, l_{4i}) \succ (c(x_{4i}), l_{4i+1}) \succ (y_{4i+1}, l_{4i+1}) & (0\leq i\leq n-1)\\
(x_{4i+1}, k_{4i+1}) : & (y_{4i+1}, l_{4i+1}) \succ (c(x_{4i+1}), l_{4i+2}) \succ (y_{4i+2}, l_{4i+2}) & (0\leq i\leq n-1)\\
(x_{4i+2}, k_{4i+2}) : & (y_{4i+3}, l_{4i+3}) \succ (c(x_{4i+2}), l_{4i+2}) \succ (y_{4i+2}, l_{4i+2}) & (0\leq i\leq n-1)\\
(x_{4i+3}, k_{4i+3}) : & (y_{4i}, l_{4i}) \succ (c(x_{4i+3}), l_{4i+3}) \succ (y_{4i+3}, l_{4i+3}) & (0\leq i\leq n-1)\\
(p_j^1, p_j^4) : & (z^1_j, z^2_j)  \succ (c_j^1, z^3_j) & (1\leq j\leq m) \vspace{1mm}\\
(p_j^2, p_j^5) : & (z^1_j, z^2_j)  \succ (c_j^2, z^4_j) & (1\leq j\leq m) \vspace{1mm}\\
(p_j^3, p_j^6) : & (z^1_j, z^2_j)  \succ (c_j^3, z^5_j) & (1\leq j\leq m) \vspace{1mm}\\ 
s_j : & c_j^1 \succ c_j^2 \succ c_j^3 & (1\leq j\leq m)\\ 
t_j : & z^3_j \succ z^4_j \succ z^5_j & (1\leq j\leq m) \vspace{1mm} \\ \\

y_{4i} : & x_{4i} \succ x_{4i+3} & (0\leq i\leq n-1)\\
y_{4i+1} : & x_{4i+1} \succ x_{4i} & (0\leq i\leq n-1)\\
y_{4i+2} : & x_{4i+1} \succ x_{4i+2} & (0\leq i\leq n-1)\\
y_{4i+3} : & x_{4i+2} \succ x_{4i+3} & (0\leq i\leq n-1)\\
l_{4i} : & k_{4i+3} \succ k_{4i} & (0\leq i\leq n-1)\\
l_{4i+1} : & k_{4i} \succ k_{4i+1} & (0\leq i\leq n-1)\\
l_{4i+2} : & k_{4i+2} \succ k_{4i+1} & (0\leq i\leq n-1)\\
l_{4i+3} : & k_{4i+3} \succ k_{4i+2} & (0\leq i\leq n-1)\\
z^1_j : & p_j^1 \succ p_j^2 \succ p_j^3 & (1\leq j\leq m) \vspace{1mm}\\
z^2_j : & p_j^6 \succ p_j^5 \succ p_j^4 & (1\leq j\leq m) \vspace{1mm}\\
z^{2+s}_j : & p_j^{3+s} \succ t_j & (1\leq j\leq m\wedge 1\leq s\leq 3) \vspace{1mm}\\
c_j^s : & p_j^s \succ x(c_j^s) \succ s_j & (1\leq j\leq m\wedge 1\leq s\leq 3) \vspace{1mm}\\

\end{array}
\]
\caption{Preference lists in the constructed instance $I$ of \hrcdual.}
\label{preflists33}
\end{figure}

The doctors in $I$ can be partitioned into two disjoint sets, $R_1 = X\cup P_1\cup S$ where $P_1 = \{ p^s_j : 1\leq s\leq 3 \}$, and $R_2 = K \cup P_2\cup T$, where $P_2 =  \{ p^s_j : 4\leq s\leq 6 \}$. Further, the hospitals in $I$ may also be partitioned into two disjoint sets, $A_1 = Y\cup Z_1\cup C_B$, where $Z_1 = \{ z^1_j : 1\leq j\leq m \}$, and $A_2 = L\cup Z_2$ where $Z_2 = \{ z^s_j : 2\leq s\leq 5 \wedge 1\leq j \leq m \}$.
For each $i\in \{1,2\}$, a doctor in $R_i$ finds acceptable only those hospitals in $A_i$ and a hospital in $A_i$ finds acceptable only those doctors in $R_i$. It follows that $I$ is an instance of \hrcdual.

\begin{figure}
\includegraphics[width = \linewidth]{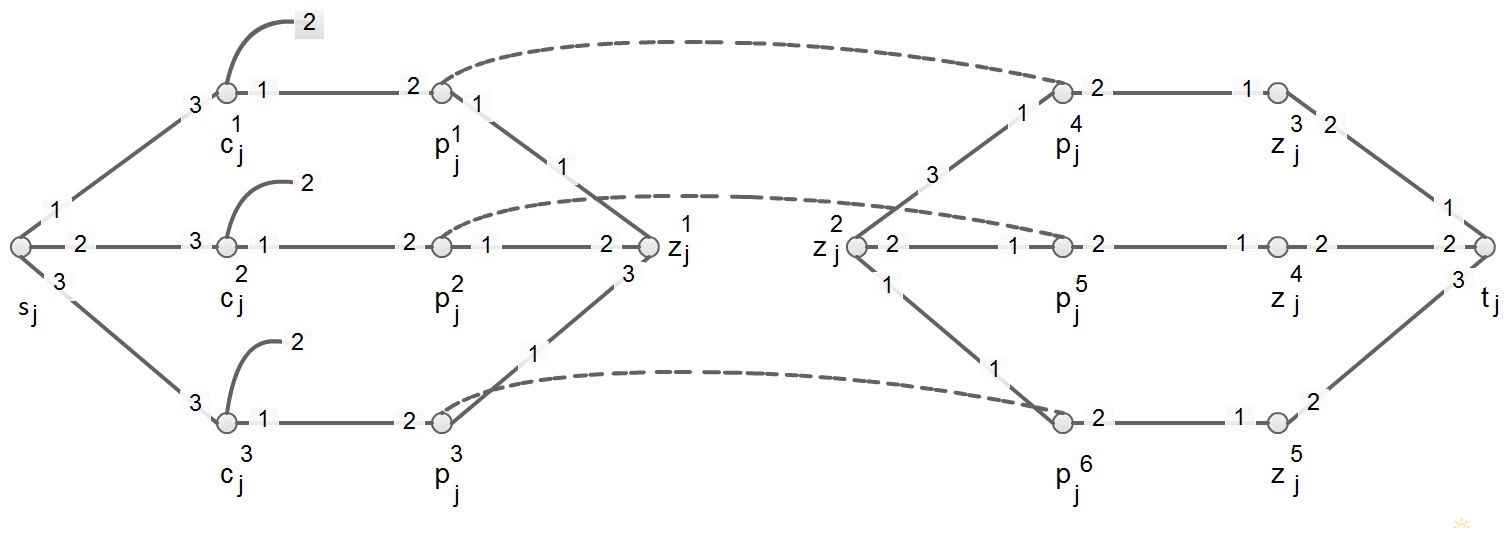}
\caption{Illustration of a clause gadget in the constructed instance $I$ of \hrcdual.}
\label{image:clausegadget}
\end{figure}

\begin{figure}
\centering
\includegraphics[scale = 0.375]{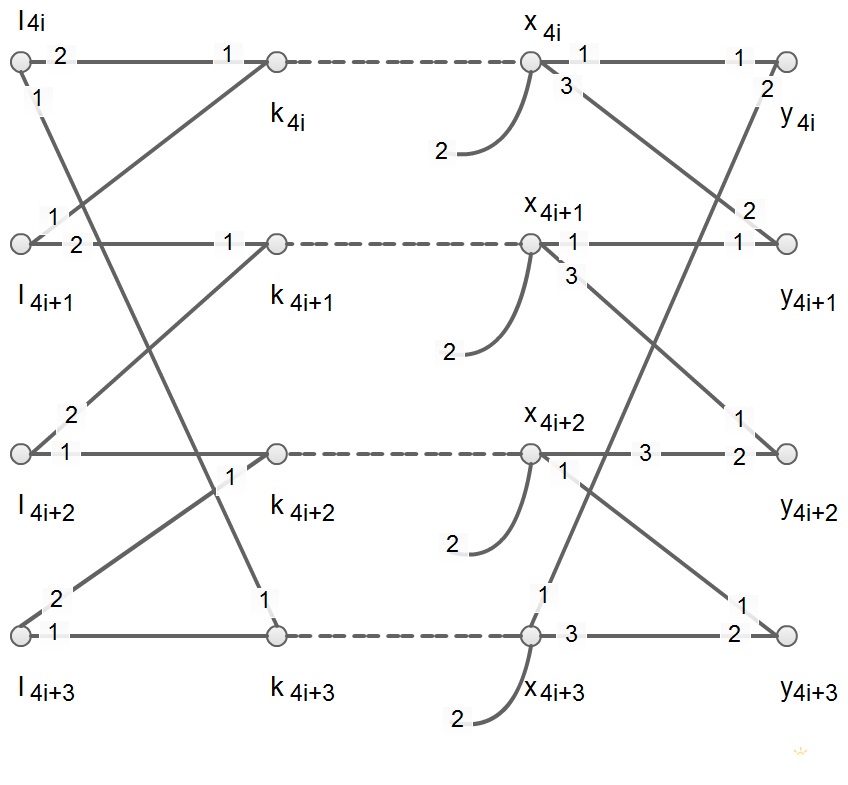}
\caption{Illustration of a variable gadget in the constructed instance $I$ of \hrcdual.}
\label{image:variablegadget}
\end{figure}

For each $i$ ($0\leq i\leq n-1$), let $T_i=\{(x_{4i+r},y_{4i+r}) : 0\leq r\leq 3\} \cup \{(k_{4i+r},l_{4i+r}) : 0\leq r\leq 3\}$ and 
$F_i=\{(x_{4i+r},y_{4i+r+1})\} : 0\leq r\leq 3\}\cup \{(k_{4i+r},l_{4i+r+1})\} : 0\leq r\leq 3\}$, where addition is taken modulo $4$.

The intuition behind the reduction is as follows.  If a variable is true, the matching $T_i$ is used in a variable gadget, otherwise the matching $F_i$ is used.  The first matching $T_i$ ensures that $x_{4i}$ and $x_{4i+1}$ (corresponding to the first and second occurrences of literal $v_i$) have their first choices.  The second matching $F_i$ ensures that $x_{4i+2}$ and $x_{4i+3}$ (corresponding to the first and second occurrences of literal $\bar{v_i}$) have their first choices.  In a complete stable matching, doctor $s_j$ must be matched, meaning that some $c_j^s$ has their third-choice partner,  This intuitively corresponds to the literal that is true in clause $c_j$.  The fact that $c_j^s$ prefers $x(c_j^s)$ to their assignee forces $x(c_j^s)$ to have their first choice, to avoid a blocking pair.  That means that the corresponding literal must be true.  The presence of the couples essentially allow two copies of the variable and clause gadgets from \cite{IMO09} to be created to allow a dual market construction, and with preference lists designed to avoid unwanted blocking pairs within these gadgets.

Formally, we will prove that $B$ is satisfiable if and only if $I$ admits a complete stable matching.  We show this through the following two claims.

\begin{claim}
    If $B$ is satisfiable then $I$ has a complete stable matching.
\end{claim}
\claimproofstart
Let $f$ be a satisfying truth assignment of $B$.  Define a complete matching $M$ in $I$ as follows.  For each variable $v_i\in V$, if $v_i$ is true under $f$, add the pairs in $T_i$ to $M$, otherwise add the pairs in $F_i$ to $M$.

Let $j$ $(1\leq j\leq m)$ be given. Then $c_j$ contains at least one literal that is true under $f$. Suppose this literal occurs at position $s$ of $c_j$  $(1\leq s \leq 3)$. Then add the pairs ($s_j, c^s_j$), ($p^s_j, z^1_j$) and ($p_j^{s+3}, z^2_j$) to $M$. For each $b\in \{1,2,3\}\setminus \{s\}$, add the pairs ($p^b_j,c^b_j$) and ($p_j^{b+3}, z_j^{b+2}$) to $M$. Finally, add the pair ($t_j, z_j^{s+2}$) to $M$.

No doctor in $S$ may form a blocking pair (since he can only potentially prefer a hospital in $C_B$, which ranks him last) nor can a doctor in $T$ (since he can only potentially prefer a hospital $z_j^s$ $(1\leq j\leq m$ and $3\leq s\leq 5)$
which ranks him last).

Suppose that ($x_{4i+r}, k_{4i+r}$) blocks $M$ with $(c(x_{4i+r}), l_{4i+a})$, where $0\leq i\leq n-1$, $0\leq r\leq 3$ and $1\leq a\leq 3$. Then ($x_{4i+r}, k_{4i+r}$) are jointly matched with their third-choice pair.
\vspace{1mm}

\noindent\emph{Case (i)}: $r\in \{0,1\}$. Then $f(v_i) = F$ because $(x_{4i+r}, y_{4i+r+1})\in M$ and $(k_{4i+r}, l_{4i+r+1})\in M$. Let $c^s_j = c(x_{4i+r})$ $(1\leq j\leq m$ and $1\leq s\leq 3)$. As $v_i$ does not make $c_j$ true then $(p_j^s, c_j^s)\in M$. Then $c_j^s$ has its first choice and cannot form part of a blocking pair,
a contradiction. 
\vspace{1mm}

\noindent\emph{Case (ii)}: $r\in \{2,3\}$ Then $f(v_i) = T$ because $(x_{4i+r}, y_{4i+r})\in M$ and $(k_{4i+r}, l_{4i+r})\in M$. Let $c^s_j = c(x_{4i+r})$ $(1\leq s\leq 3$ and $1\leq j\leq m)$. As $\bar{v}_i$ does not make $c_j$ true then $(p_j^s, c_j^s)\in M$. This means that $c_j^s$ has its first choice and cannot form part of a blocking pair,
a contradiction. 
\vspace{1mm}

Suppose that $(x_{4i+r},k_{4i+r})$ blocks $M$ with hospitals in $Y\cup L$ for some $i$ ($0\leq i\leq n-1$) and $r$ ($0\leq r\leq 3$).  If $T_i\subseteq M$ then $r\in \{2,3\}$ so that $(x_{4i+r},k_{4i+r})$ have their third choice.  But then, each of $y_{4i}$ and $l_{4i+3}$ have their first choice, so $(x_{4i+r},k_{4i+r})$ cannot block $M$ after all, a contradiction.  If $F_i\subseteq M$ then $r\in \{0,1\}$, so $(x_{4i+r},k_{4i+r})$ have their third choice.  But then, each of $l_{4i}$ and $l_{4i+1}$ have their first choice, so $(x_{4i+r},k_{4i+r})$ cannot block $M$ after all, a contradiction.  

Now suppose $(p^s_j, p_j^{s+3})$ blocks $M$. Then $(p^s_j, c^s_j)\in M$ and $(p_j^{s+3}, z_j^{s+2})\in M$, and $(p_j^s, p_j^{s+3})$ jointly prefer $(z^1_j, z^2_j)$ to their partners. At most one of $\{z^1_j, z^2_j\}$  can prefer the relevant member of $\{p^s_j, p^{s+3}_j\}$ to their partner, whilst the other would prefer their current partner in $M$ and thus could not form part of a blocking pair, a contradiction. Hence, $M$ is a complete stable matching in $I$.
\claimproofend
\begin{claim}
    If $I$ has a complete stable matching, then $B$ is satisfiable.
\end{claim}
\specialclaimproofstart
Conversely, suppose that $M$ is a complete stable matching in $I$. We form a truth assignment $f$ in $B$ as follows.  For each $i$ ($0\leq i\leq n-1$), 
$M\cap ( (X_i\times Y_i)\cup (K_i\times L_i) )$ is a perfect matching of  ($X_i\cup Y_i) \cup (K_i\cup L_i$). If $M\cap ( (X_i\times Y_i)\cup (K_i\times L_i))=T_i$, set $v_i$ to be true under $f$.  Otherwise, $M\cap ( (X_i\times Y_i)\cup (K_i\times L_i))=F_i$, in which case we set $v_i$ to be false under $f$.

Now, let $c_j$ be a clause in $C_B$ ($1\leq j\leq m$).  There exists some $s$ ($1\leq s\leq 3$) such that $(s_j,c_j^s)\in M$.  Let $x_{4i+r}=x(c_j^s)$, for some $i$ ($0\leq i\leq n-1$) and $r$ ($0\leq r\leq 3$).  If $r\in \{0,1\}$, then $(x_{4i+r},y_{4i+r})\in M$ (or equivalently $(k_{4i+r}, l_{4i+r})\in M$) by the stability 
of $M$. Thus, variable $v_i$ is true under $f$, and hence clause $c_j$ is true under $f$, since literal $v_i$ occurs in $c_j$.  If $r\in \{2,3\}$, then $(x_{4i+r},y_{4i+r+1})\in M$ (or equivalently $(k_{4i+r}, l_{4i+r+1})\in M$) (where addition is taken modulo $4$) by the stability of $M$.  Thus variable $v_i$ is false under $f$, and hence clause $c_j$ is true under $f$, since literal $\bar{v_i}$ occurs in $c_j$.  Hence $f$ is a satisfying truth assignment of $B$.
\specialclaimproofend
\end{proof}

%
%
%

Let $I$ be the instance of \hrcdual\ as constructed in the proof of Theorem \ref{33-COM-HRC} from an instance of \sat. We now add additional doctors and hospitals to $I$, in the form of a series of \emph{enforcer gadgets}, to obtain a new instance $I'$ of \hrcdual\ as follows, with the aim that every stable matching in $I'$ must be complete.  For each $i$ ($0\leq i\leq n-1$) and $r$ ($0\leq r\leq 3$), we add an enforcer gadget to $y_{4i+r}$ comprising two new couples,
$(u_{4i+r}^1,u_{4i+r}^2)$ and $(u_{4i+r}^3,u_{4i+r}^4)$, a single doctor $u_{4i+r}^5$, and four new hospitals $h_{4i+r}^s$ ($1\leq s\leq 4$), all with capacity 1.
The preference lists of the newly added agents corresponding to $y_{4i+r}\in Y$ are shown in Figure \ref{preflists:addedgadget} and the enforcer gadget is illustrated in Figure \ref{image:addedgadget}.  In the preference list of $y_{4i+r}$ in $I'$, $P(y_{4i+r})$ represents the preference list of $y_{4i+r}$ in $I$.

\begin{figure}
\[
\begin{array}{rll}
(u^1_{4i+r}, u^2_{4i+r}) : & (h^1 _{4i+r}, h^2 _{4i+r}) ~~ ~~ & (0\leq i\leq n-1, 0\leq r\leq 3)\\
(u^3_{4i+r}, u^4_{4i+r}) : & (h^1 _{4i+r}, h^4 _{4i+r}) \succ (h^3 _{4i+r}, h^2 _{4i+r})~~ & (0\leq i\leq n-1, 0\leq r\leq 3)\\
u^5 _{4i+r} : & y_{4i+r} \succ h^1 _{4i+r} ~~ & (0\leq i\leq n-1, 0\leq r\leq 3) \\ \\

y_{4i+r} : & P(y_{4i+r}) \succ u^5 _{4i+r} & (0\leq i\leq n-1, 0\leq r \leq 3) \vspace{1mm} \\
h^1_{4i+r} : & u^5 _{4i+r} \succ u^1 _{4i+r} \succ u^3 _{4i+r} & (0\leq i\leq n-1, 0\leq r \leq 3) \vspace{1mm}\\
h^2_{4i+r} : & u^4 _{4i+r} \succ u^2 _{4i+r} & (0\leq i\leq n-1, 0\leq r \leq 3) \vspace{1mm}\\
h^3_{4i+r} : & u^3 _{4i+r} & (0\leq i\leq n-1, 0\leq r \leq 3) \vspace{1mm}\\
h^4_{4i+r} : & u^4 _{4i+r} & (0\leq i\leq n-1, 0\leq r \leq 3) \vspace{1mm}\\

\end{array}
\]
\caption{Preference lists of agents in the enforcer gadget corresponding to $y_{4i+r}$ in the constructed instance $I'$ of \hrcdual.}
\label{preflists:addedgadget}
\end{figure}
\begin{figure}[htbp]
\[
\includegraphics[scale = 0.35]{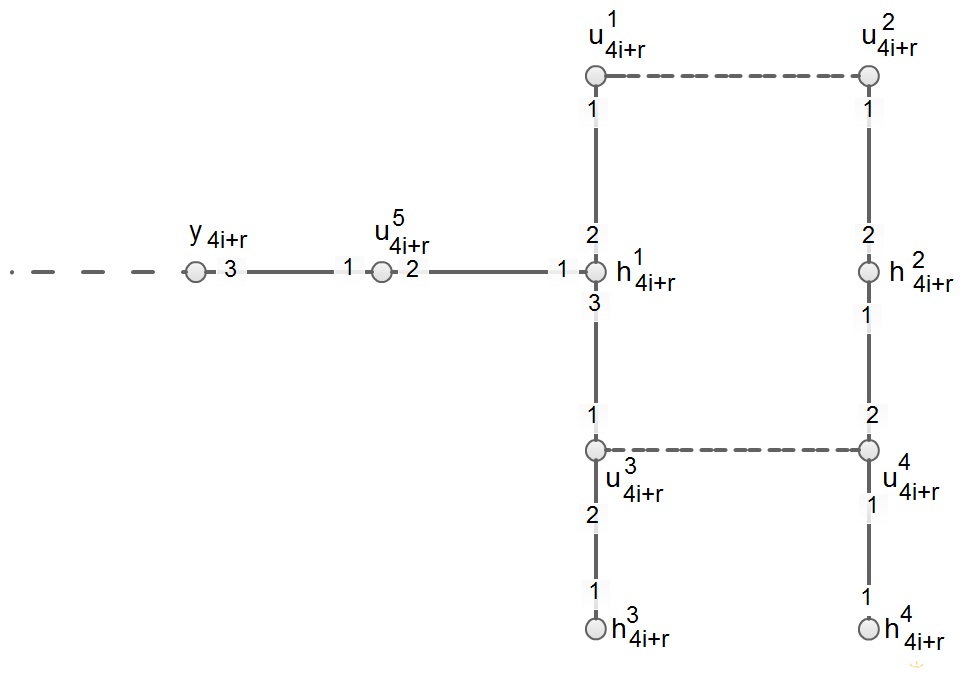}
\]
\caption{Illustration of the enforcer gadget corresponding to $y_{4i+r}$ in the constructed instance $I'$ of \hrcdual.}
\label{image:addedgadget}
\end{figure}

The enforcer gadget corresponding to $y_{4i+r}\in Y$ forces $y_{4i+r}$ to have a partner in $X$ in any stable matching in $I'$, as the following lemma shows.
\begin{lemma}
\label{y4i+rmustbematched}
In any stable matching $M'$ in $I'$, for every $y_{4i+r} \in Y$, $M'(y_{4i+r}) \in X$.  
\end{lemma}
\begin{proof}
Suppose not. Then, $y_{4i+r}$ is either unmatched in $M'$ or $M'(y_{4i+r}) = u^5_{4i+r}$. If $y_{4i+r}$ is unmatched then $( u^5 _{4i+r}, y_{4i+r})$ blocks $M'$ in $I'$, a contradiction. Hence,  $M'(y_{4i+r}) = u^5_{4i+r}$ and thus $(u^5_{4i+r}, h^1_{4i+r}) \notin M$. 

Assume that $h^1_{4i+r}$ is unmatched in $M'$. Then $(u^1_{4i+r}, u^2_{4i+r})$ is unmatched in $M'$ and also $M'((u^3_{4i+r}, u^4_{4i+r})) \neq (h^1_{4i+r}, h^4_{4i+r})$. Thus, either $M'((u^3_{4i+r}, u^4_{4i+r}) ) = (h^3_{4i+r}, h^2_{4i+r})$ or $(u^3_{4i+r}, u^4_{4i+r})$ is unmatched in $M'$. If $(u^3_{4i+r}, u^4_{4i+r})$ is unmatched in $M'$ then $(u^1_{4i+r}, u^2_{4i+r})$ blocks $M'$ with $(h^1_{4i+r}, h^2_{4i+r})$ in $I'$, a contradiction. Hence $(u^3_{4i+r}, u^4_{4i+r})$ must be matched with $(h^3_{4i+r},h^2_{4i+r})$ in $M'$. However, now $(u^3_{4i+r}, u^4_{4i+r})$ blocks $M'$ with $(h^1_{4i+r}, h^4_{4i+r})$ in $I'$, a contradiction. 

Now, assume $h^1_{4i+r}$ is matched in $M'$. Further assume that $h^1_{4i+r}$ is matched in $M'$ through the joint matching of $(u^1_{4i+r}, u^2_{4i+r})$ with $(h^1 _{4i+r},$ $ h^2 _{4i+r})$. However, in this case, $M'$ is blocked by $(u^3_{4i+r}, u^4_{4i+r})$ with $(h^3 _{4i+r}, h^2 _{4i+r})$ in $I'$, a contradiction. Thus, $h^1_{4i+r}$ must be matched in $M'$ through the joint matching of $(u^3_{4i+r}, u^4_{4i+r})$ with $(h^1 _{4i+r},$ $ h^4 _{4i+r})$. However, in this case, $M'$ is blocked by $(u^1_{4i+r}, u^2_{4i+r})$ with $(h^1 _{4i+r}, h^2 _{4i+r})$ in $I'$, a contradiction.
\end{proof}

We now show through the following two lemmata that if $M^{\prime }$ is a stable matching in $I'$, then $M=M'\setminus (U\times H)$
is a complete stable matching in $I$.
\begin{lemma}
\label{Z-OR-PUNIONT}
No hospital in $Z\cup C_B$ may be unmatched and no doctor in $P\cup S\cup T$ may be unmatched in any stable matching $M'$ in $I'$.
\end{lemma}

\begin{proof}
Assume $z^1_j$ is unmatched in $M'$ for some $j$ $(1\leq j \leq m)$. Thus $(p_j^s, z^2_j) \notin M$  $(4\leq s\leq 6)$ as $z^2_j$ must also be unmatched. Hence, $(z^1_j, z^2_j)$ are jointly unmatched and find $(p^1_j, p^4_j)$ acceptable. Furthermore, $(p^1_j, p^4_j)$ jointly prefers $(z^1_j, z^2_j)$ to any other pair. Hence, $(z^1_j, z^2_j)$ blocks $M$ with $(p^1_j, p^4_j)$, a contradiction. Thus, $z^1_j$ must be matched in any stable matching admitted by $I'$. By a similar argument, the same holds for $z^2_j$.

Assume $t_j'$  is unmatched in $M'$ for some $j$ $(1\leq j \leq m)$. If some $z^s_j$ $(3\leq s \leq 5)$ is unmatched, then $(t^s_j , z^s_j)$ blocks $M'$ in $I'$, a contradiction. Thus, $\{(p^4_j, z^3_j), (p^5_j, z^4_j), (p^6_j, z^5_j) \}\subseteq M'$. Then $z^2_j$ is unmatched in $M'$, a contradiction. Thus, $t_j$ must be matched in any stable matching in $I'$.


Assume some doctor $p^s_j$ is unmatched in $M'$ for some $j$ $(1\leq j \leq m)$ and $s$ $(1\leq s\leq 3)$. Then $(p^s_j, p^{s+3}_j)$ is unmatched in $M'$. Hence, $(p^s_j, p^{s+3}_j)$ blocks $M'$ in $I'$ with $(c^s_j, z^{s+2}_j)$, a contradiction. Thus, all doctors in $P$ must be matched in $M'$.

Assume some $z^s_j$ $(1\leq j\leq m$ and $3\leq s\leq 5)$ is unmatched in $M'$.  We have already shown that all doctors in $P\cup T$ are matched in $M'$.  Hence $(p_j^{s+1},z_j^2)\in M'$ and $(t_j,z_j^b)\in M'$ for some $b\in \{3,4,5\} \setminus \{s\}$.  It then follows that $p_j^{b+1}$ is unmatched in $M'$, a contradiction.  Thus all hospitals in $Z$ must be matched in $M'$.

Observe that no $c^s_j$ $(1\leq s\leq 3, 1\leq j\leq m)$ can be matched to $x(c^s_j)=x_{4i+r}$ in $M'$, for otherwise $M'(y_{4i+r'}) \notin X$ for some $r'$ ($0\leq r'\leq 3$), a contradiction to Lemma \ref{y4i+rmustbematched}. Since $z^1_j$ must be matched in $M'$ to some $p^s_j$ $(1\leq s\leq 3)$, and since no doctor in $P$ may be unmatched in $M'$, then for each $s^{\prime } \in \{1,2,3\}\setminus \{s\}$, $(p^{s^{\prime }}_j,c^{s^{\prime }}_j)\in M'$.  Thus, $(s_j, c^s_j) \in M'$ for otherwise $(s_j,c^s_j)$ blocks $M'$ in $I'$. Thus all doctors in $S$ and hospitals in $C_B$ must be matched in $M'$.
\end{proof}

\begin{lemma}
\label{XUNIONKORYUNIONL}
No hospital in $L\cup Y$ may be unmatched and no doctor in $K\cup X$ may be unmatched in any stable matching $M'$ in $I'$.
\end{lemma}

\begin{proof} 
By Lemma \ref{y4i+rmustbematched}, $M'(y_{4i+r}) \in X$ for all $y_{4i+r} \in Y$. Hence $M'(x_{4i+r}) \in Y$ for all $x_{4i+r} \in X$. Since $(x_{4i+r} , k_{4i+r})$ are a couple for all $x_{4i+r} \in X$, it follows that $M'(k_{4i+r}) \in L$ for all $k_{4i+r} \in K$ and $M'(l_{4i+r}) \in K$ for all $l_{4i+r} \in L$.
\end{proof}

The previous three lemmata now allow us to state the following more general theorem.

\begin{theorem}
\label{33comdmml}
\hrcdual\ is NP-hard, even if all preference lists have length 3, all hospitals have capacity 1, and the preference list of each single doctor, couple and hospital is derived from a strictly ordered master list of hospitals, pairs of hospitals and doctors, respectively.
\end{theorem}

\begin{proof}
Let $B$ be an instance of  \sat. Construct an instance $I$ of \hrcdual\ as described in the proof of Theorem \ref{33-COM-HRC} and extend this instance as described above by adding the enforcer gadgets corresponding to the doctors in $U$ and the hospitals in $H$, to obtain the instance $I'$ of \hrcdual.

Let $(R_1,A_1)$ and $(R_2,A_2)$ be the dual market as established in the proof of Theorem \ref{33-COM-HRC}.  Then $I'$ is an instance of \hrcdual\ with dual market $(R_1',A_1')$ and $(R_2',A_2')$ if we define
\[
\begin{array}{rl}
R_1'=R_1\cup U_1, & \mbox{where }U_1 = \{ u^s_{4i+r} : s \in \{1, 3, 5\} \}\\
R_2'=R_2\cup U_2, & \mbox{where }U_2 = \{ u^s_{4i+r} : s \in \{2, 4\} \}\\
A_1'=A_1\cup H_1', & \mbox{where }H_1' = \{ h^s_{4i+r} : 0\leq i\leq n-1\wedge 0 \leq r\leq 3\wedge s \in \{ 1,3 \} \}\\
A_2'=A_2\cup H_2', & \mbox{where }H_2 = \{ h^s_{4i+r} : 0\leq i\leq n-1\wedge 0 \leq r\leq 3\wedge s \in \{ 2,4 \} \}.
\end{array}
\]

Let $f$ be a satisfying truth assignment of $B$. Define a matching $M$ in $I$ as in the proof of Theorem \ref{33-COM-HRC}. Define a matching $M^{\prime }$ in $I'$ as follows:
$$M^{\prime } = M \cup \{ (u^5_{4i+r}, h^1_{4i+r}),(u^3_{4i+r}, h^3_{4i+r}), (u^4_{4i+r}, h^2_{4i+r}) : 0\leq i \leq n-1\wedge 0\leq r\leq 3 \}.$$

As shown in the proof of Theorem \ref{33-COM-HRC}, no agent outside of $Y\cup U\cup H$ can block $M^{\prime }$ in $I$. 
For each $i$ ($0\leq i\leq n-1$) and $r$ ($0\leq r\leq 3$), since $h_{4i+r}^1$ has its first-choice doctor in $M'$, the couple $(u^1_{4i+r}, u^2_{4i+r})$ cannot block $M'$ in $I'$.  For a similar reason, the couple $(u^3_{4i+r}, u^4_{4i+r})$ cannot block $M'$ in $I$.  Finally, $M'(y_{4i+r})\in X$, hence $y_{4i+r}$ cannot block $M'$ in $I'$.  Thus, $M'$ is stable in $I'$.
\medskip

Conversely, suppose that $M^{\prime }$ is a stable matching in $I'$. By Lemma \ref{y4i+rmustbematched}, every $y_{4i+r}\in Y$ is matched in $M^{\prime }$ to a doctor in $X$. By Lemmas \ref{Z-OR-PUNIONT} and \ref{XUNIONKORYUNIONL}, every doctor in $X\cup K\cup P\cup S\cup T$ is matched in $M^{\prime}$.  Now let $M=M'\setminus (U\times H)$.
Then $M$ is a complete matching in $I$.  Moreover $M$ must be stable in $I$, otherwise we obtain a blocking pair of $M'$ in $I'$, a contradiction.
By the proof of Theorem \ref{33-COM-HRC} we can obtain a satisfying truth assignment for $B$ from $M$.

Finally, the master lists shown in Figures \ref{masterlistrescouples},  \ref{masterlistdoctors} and \ref{masterlisthospitals} indicate that the preference list of each single doctor, couple and hospital may be derived from a master list of hospital pairs, doctors and hospitals, respectively (for brevity, in the master list in each of these figures, the symbol $\succ$ is omitted from between consecutive entries). 
\end{proof}

\begin{landscape}
\begin{figure}
\[
\begin{array}{c}
L^1_i :
(y_{4i}, l_{4i}) ~
(c(x_{4i}), l_{4i+1}) ~
(y_{4i+1}, l_{4i+1}) ~
(c(x_{4i+1}), l_{4i+2}) ~ 
(c(x_{4i+3}), l_{4i+3}) ~
(y_{4i+3}, l_{4i+3}) ~
(c(x_{4i+2}), l_{4i+2}) ~
(y_{4i+2}, l_{4i+2}) ~~~~ (0 \leq i \leq n-1)\\
\\
L^2_{i,r} :
(h^1_{4i+r}, h^2_{4i+r}) ~
(h^1_{4i+r}, h^4_{4i+r}) ~
(h^3_{4i+r}, h^2_{4i+r}) ~~~~ (0 \leq i \leq n-1, 0 \leq r \leq 3)

~~~~~~~~~

L^3_j : 
(z^1_j, z^2_j) ~
(c^1_j, z^3_j) ~
(c^2_j, z^4_j) ~
(c^3_j, z^5_j) ~~~~ (1 \leq j \leq m)
\\
\\
\mbox{Master List} : L^1_0 ~ L^1_1 ~ \dots ~ L^1_{n-1}~ L^2_{0,0} ~ L^2_{0,1} ~ L^2_{0,2} ~ L^2_{0,3} ~ L^2_{1,0} ~ L^2_{1,1} ~ L^2_{1,2} ~ L^2_{1,3} ~ \ldots ~ L^2_{(n-1),0} ~ L^2_{(n-1),1} ~ L^2_{(n-1),2} ~ L^2_{(n-1),3} ~  L^3_1 ~ L^3_2 ~ \dots ~ L^3_m

\end{array}
\]
\caption{Master list of preferences for couples in the \hrcdual\ instance $I'$.}
\label{masterlistrescouples}
\end{figure}

\begin{figure}
\[
\begin{array}{c}
L^4_j :
c^1_j ~
c^2_j ~
c^3_j ~
z^3_j ~
z^4_j ~
z^5_j

~~~~ (1 \leq j \leq m)

~~~~~~~~

L^5_{i,r} :

y_{4i+r} ~
h^1_{4i+r} ~~~~ (0 \leq i \leq n-1, 0 \leq r \leq 3)
\\
\\
\mbox{Master List} : L^4_1 ~ L^4_2 ~ \dots ~ L^4_m ~ L^5_{0,0} ~ L^5_{0,1} ~ L^5_{0,2} ~ L^5_{0,3} ~ L^5_{1,0} ~ L^5_{1,1} ~ L^5_{1,2} ~ L^5_{1,3} ~ \ldots ~ L^5_{(n-1),0} ~ L^5_{(n-1),1} ~ L^5_{(n-1),2} ~ L^5_{(n-1),3}
\end{array}
\]
\caption{Master list of preferences for single doctors in the \hrcdual\ instance $I'$.}
\label{masterlistdoctors}
\end{figure}

\begin{figure}
\[
\begin{array}{c}
L^6_i : 
x_{4i+1} ~
x_{4i} ~
x_{4i+2} ~
x_{4i+3} ~
k_{4i+3} ~
k_{4i} ~
k_{4i+2} ~
k_{4i+1} 
~~~~ (0 \leq i \leq n-1)
\\
\\
L^7_j : 
p^1_j ~
p^2_j ~
p^3_j ~
p^6_j ~
p^5_j ~
p^4_j 
~~~~ (1 \leq j \leq m)
~~~~~~~~
L^8 : s_1 ~ s_2 ~\ldots ~ s_m ~ t_1 ~ t_2 ~ \ldots ~ t_m 
\\
\\

L^9_{i,r} :
u^5_{4i+r} ~
u^1_{4i+r} ~
u^3_{4i+r} ~
u^4_{4i+r} ~
u^2_{4i+r} ~~~~ (0 \leq i \leq n-1, 0 \leq r \leq 3)
\\
\\
\mbox{Master List} : L^7_1 ~ L^7_2 ~ \dots ~ L^7_m  ~ L^6_0 ~ L^6_1 ~ \dots ~ L^6_{n-1} ~ L^8 ~ L^9_{0,0} ~ L^9_{0,1} ~ L^9_{0,2} ~ L^9_{0,3} ~ L^9_{1,0} ~ L^9_{1,1} ~ L^9_{1,2} ~ L^9_{1,3} ~ \ldots ~ L^9_{(n-1),0} ~ L^9_{(n-1),1} ~ L^9_{(n-1),2} ~ L^9_{(n-1),3}
\end{array}
\]
\caption{Master list of preferences for hospitals in the \hrcdual\ instance $I'$.}
\label{masterlisthospitals}
\end{figure}
\end{landscape}

\subsection{Matchings with the minimum number of blocking pairs}
\label{sec:inapprox}
In this section we show that {\sc min bp hrc} is NP-hard and very difficult to approximate, even in the unit hospital capacity case and where each couple finds acceptable only one hospital pair.

\begin{theorem}
\label{thm:min-bp-approx}
{\sc min bp hrc} is not approximable within $m^{1-\varepsilon}$, for any $\varepsilon>0$, where $m$ is the total length of the hospitals' preference lists, unless \emph{P=NP}.  The result holds even if there are no single doctors, each couple is of type-a and finds acceptable only one pair of hospitals, and each hospital has capacity 1.
\end{theorem}
\begin{proof}  
Let $\varepsilon>0$ be given.  We reduce from {\sc com smti}, which is shown to be NP-complete in the Appendix of \cite{MM10} even if every man's list is strictly ordered and of length at most 3, and every woman's list is either strictly ordered and of length at most 3, or is a tie of length 2.  Hence, let $I$ be an instance of this restriction of {\sc com smti}.

Let $U$ be the set of men in $I$ and let $W$ be the set of women in $I$.  Let $G=(U\cup W,E)$ denote the underlying bipartite graph of $I$, and let $m_I$ denote $|E|$.  Let $W^s$ denote the set of women in $W$, whose preference lists are strictly ordered in $I$, and let $W^t=W\setminus W^s$.  Then, every woman in $W^t$ has a preference list in $I$ that is a tie of length 2.  Let $n_t=|W^t|$, $C=\lceil \frac{3}{\varepsilon}\rceil$ and $B=n_t m_I^C+1$.

We construct an instance $J$ of \hrc\ on the basis of $I$ as follows.  Let $H$ be the set of hospitals in $J$, where $H=H_U\cup H_W$, $H_U=\{h_u : u\in U\}$ and $H_W=\{h_w : w\in W\}$.  Let $H_W^s=\{h_w : w\in W^s\}$, and let $H_W^t=H_W\setminus H_W^s$.

For each person $p\in U\cup W$, let $\succeq_p$ denote $p$'s preference list in $I$.  Each hospital $h_p\in H$ is initially given the preference list $\succeq_p$ in $J$.  Note that if $p\in W^t$, then the tie of length 2 is broken arbitrarily in $\succeq_p$, so that we obtain a strictly ordered ranking $\succ_p$ for every hospital $h_p$ in $J$.

Now let $e=(u,w)\in E$.  Suppose firstly that $w\in W^t$.  Create the couple $C^e=(c_u^e,c_w^e)$ who find acceptable the hospital pair $(h_u,h_w)$.  Replace $u$ by $c_u^e$ in $\succ_{h_w}$, and replace $w$ by $c_w^e$ in $\succ_{h_u}$.  Now suppose that $w\in W^s$.  Create the couple $C^{e,k}=(c_u^{e,k},c_w^{e,k})$ who find acceptable the hospital pair $(h_u,h_w)$, for each $k$ ($1\leq k\leq B$).  Replace $u$ by $c_u^{e,1}\succ_{h_w}\dots \succ_{h_w} c_u^{e,B}$ in $\succ_{h_w}$, and replace $w$ by $c_w^{e,1}\succ_{h_u}\dots \succ_{h_u} c_w^{e,B}$ in $\succ_{h_u}$.  Let $P(h_w)$ denote $\succ_{h_w}$ at this point.

Next, for each hospital $h_w\in H_W$, add two additional hospitals $h_w'$ and $h_w''$, and add $3B$ couples $(c_{w,1}^k,c_{w,2}^k)$, $(c_{w,3}^k,c_{w,4}^k)$ and $(c_{w,5}^k,c_{w,6}^k)$ ($1\leq k\leq B$).  The preference lists of $h_w$ and the newly added couples and hospitals are shown in Figure \ref{fig:hrc-prefs}; in particular, $\succ_{h_w}$ is extended with the addition of $c_{w,1}^k$ and $c_{w_6}^k$ ($1\leq k\leq B$).
\begin{figure}[t!]
\[
\begin{array}{rll}
(c_{w,1}^k,c_{w,2}^k) : & (h_w,h_w')  & \hspace{-2cm} (h_w\in H_W, 1\leq k\leq B) \vspace{5pt} \\
(c_{w,3}^k,c_{w,4}^k) : & (h_w',h_w'') & \hspace{-2cm} (h_w\in H_W, 1\leq k\leq B)  \vspace{5pt} \\
(c_{w,5}^k,c_{w,6}^k) : & (h_w'',h_w) & \hspace{-2cm} (h_w\in H_W, 1\leq k\leq B)  \vspace{5pt} \\
\\
h_w : & P(h_w)\succ c_{w,1}^1\succ\dots\succ c_{w,1}^B\succ c_{w,6}^1\succ\dots\succ c_{w,6}^B ~~~~ & (h_w\in H_W) \vspace{5pt} \\
h_w' : & c_{w,3}^1\succ\dots\succ c_{w,3}^B\succ c_{w,2}^1\succ\dots\succ c_{w,2}^B & (h_w\in H_W) \vspace{5pt} \\
h_w'' : & c_{w,5}^1\succ\dots \succ c_{w,5}^B\succ c_{w,4}^1\succ\dots\succ c_{w,4}^B & (h_w\in H_W) \vspace{5pt}
\end{array}
\]
\caption{Preference lists in the constructed instance of {\sc min bp hrc}.}
\label{fig:hrc-prefs}
\end{figure}

Every hospital in $J$ has capacity 1.  It is easy to see that there are no single doctors in $J$, and each couple is of type-a and finds acceptable only one pair of hospitals.  We claim that (i) if $I$ admits a complete stable matching then $J$ admits a matching with at most $n_t$ blocking pairs, and (ii) if $I$ admits no complete stable matching then any matching in $J$ has at least $B$ blocking pairs.

To show (i), suppose that $M$ is a complete stable matching in $I$.  Construct a 
 matching $M'$ in $J$ as follows.  For each pair $(u,w)$ in $M'$, representing the underlying edge $e$ in $G$, suppose firstly that $w\in W^t$.  Match the couple $(c_u^e,c_w^e)$ with the hospital pair $(h_u,h_w)$ in $M'$.  Now suppose that $w\in W^s$.  Match the couple $(c_u^{e,1},c_w^{e,1})$ with the hospital pair $(h_u,h_w)$ in $M'$. Finally, for each $h_w\in H_W$ and $k$ ($1\leq k\leq B$), match the couple $(c_{w,3}^k,c_{w,4}^k)$ with the hospital pair $(h_w',h_w'')$.  Each hospital is then full in $M'$, since $M$ is a complete stable matching in $I$.

For each $h_w\in H_W$, and $k$ ($1\leq k\leq B$), it may be verified that $(c_{w,1}^k,c_{w,2}^k)$ cannot be involved in blocking $M'$ in $J$ , since $h_w$ is assigned to a better doctor than $c_{i,1}^1$ in $M'$.  Similarly $(c_{w,5}^k,c_{w,6}^k)$ cannot be involved in blocking $M'$ in $J$ for the same reason. Since $M$ is a complete stable matching in $I$, the only possible blocking pairs of $M'$ in $J$ involve couples of the form $(c_u^e,c_w^e)$, where $w\in W^t$.  Since each woman in $W^t$ has a preference list of length 2 in $J$, there can be at most one blocking pair for each $w\in W^t$, and hence $|bp(J,M')|\leq n_t$.

To show (ii), suppose that $I$ admits no complete stable matching, and let $M'$ be any matching in $J$.  Suppose firstly that some hospital $h_w$ is unassigned in $M'$ ($h_w\in H_W$).  Then $(c_{w,5}^k,c_{w,6}^k)$ is unassigned in $M'$ for all $k$ ($1\leq k\leq B$).  Hence, $(c_{w,5}^k,c_{w,6}^k)$ blocks $M'$ with $(h_w'',h_w)$ for all $k$ ($1\leq k\leq B$), so $|bp(J,M')|\geq B$.  Now suppose that $(c_{w,6}^k,h_w)\in M'$ for some $h_w\in H_W$ and $k$ ($1\leq k\leq B)$.  Then, $(c_{w,5}^k,h_w'')\in M'$.  Hence, $(c_{w,1}^{k'},c_{w,2}^{k'})$ and $(c_{w,3}^{k'},c_{w,4}^{k'})$ are unassigned in $M'$ for all $k'$ ($1\leq k'\leq B)$ and also $h_w'$ is unassigned in $M'$.  Therefore, $(c_{w,1}^{k'},c_{w,2}^{k'})$ blocks $M'$ with $(h_w,h_w')$ for all $k'$ ($1\leq k'\leq B$), so $|bp(J,M')|\geq B$.  Finally, suppose that $(c_{w,1}^k,h_w)\in M'$ for some $h_w\in H_W$ and $k$ ($1\leq k\leq B)$.  Then, $(c_{w,2}^k,h_w')\in M'$.  Hence, $(c_{w,3}^{k'},c_{w,4}^{k'})$ and $(c_{w,5}^{k'},c_{w,6}^{k'})$ are unassigned in $M'$ for all $k'$ ($1\leq k'\leq B)$ and also $h_w''$ is unassigned in $M'$. This implies that $(c_{w,3}^{k'},c_{w,4}^{k'})$ blocks $M'$ with $(h_w',h_w'')$ for all $k'$ ($1\leq k'\leq B$), so $|bp(J,M')|\geq B$.   

We thus assume that each $h_w\in H_W$ has an assignee better than $c_{w,1}^1$ in $M'$.  Let $M$ be the matching in $I$ obtained by replacing each assignment of a couple $(c_u^e,c_w^e)$ with the pair $(u,w)$ in $M$, and by replacing each assignment of a couple $(c_u^{e,k},c_w^{e,k})$ with the pair $(u,w)$ in $M$.  It is easy to see that $M$ is a matching in $I$, since each of $(c_u^e,c_w^e)$ and $(c_u^{e,k},c_w^{e,k})$ finds acceptable only $(h_u,h_w)$, and each hospital has capacity 1 in $J$.  Moreover, $M$ is a complete matching in $I$, since in $M'$, each hospital $h_w\in H_W$ has an assignee better than $c_{w,1}^1$.  Hence by assumption, $M$ is not stable in $I$, so $M$ admits a blocking pair $(u,w)$ in $I$.

Let $e=(u,w)\in E$.  Since the preference list in $I$ of each woman in $W^t$ is a tie of length 2, and $M$ is complete, $w\in W^s$.  Since $(u,w)\notin M$, $(c_u^{e,k},c_w^{e,k})$ is unassigned in $M'$ for each $k$ ($1\leq k\leq B)$.  As $(u,w)$ blocks $M$ in $I$, $w\succ_u M(u)$ and $u\succ_w M(w)$.  Hence by construction, in $J$, $c_w^{e,k}\succ_{h_u} M(h_u)$ and $c_u^{e,k}\succ_{h_w} M(h_w)$.  It follows that $|bp(J,M')|\geq B$.

We have established that if $I$ admits a complete stable matching then $J$ admits a matching with at most $n_t$ blocking pairs.  An $m_I^C$-approximation algorithm for {\sc min bp hrc} will thus returns a matching in $J$ with at most $n_t m_I^C$ blocking pairs.  On the other hand if $I$ admits no complete stable matching then any matching in $J$ has at least $B=n_t m_I^C+1$ blocking pairs, so any approximation algorithm for {\sc min bp hrc} will return a matching in $J$ with at least $n_t m_I^C +1$ blocking pairs.

Hence, the existence of an $m_I^C$-approximation algorithm for {\sc min bp hrc} implies a polynomial-time algorithm for {\sc com smti} in a special case that is NP-complete, a contradiction unless P=NP.  Let $m$ denote the total length of the hospitals' preference lists in $J$.  We claim that $m_I^C\geq m^{1-\varepsilon}$, and then the proof is complete.

We firstly obtain the following upper bound for $m$:
\begin{eqnarray}
m\leq Bm_I + 6B|W|\leq 7Bm_I = 7m_I(n_t m_I^C + 1)\leq 7m_I^{C+2},
\label{ineq1}
\end{eqnarray}
where the last inequality follows since $n_t < m_I$. 
 We lose no generality by assuming that $m_I\geq 7$, in which case we can easily obtain the following lower bound for $m$:
\begin{eqnarray}
m\geq 6B|W|\geq B=n_t m_I^C+1\geq m_I^C\geq 7^C.
\label{ineq2}
\end{eqnarray}
Inequalities \ref{ineq1} and \ref{ineq2} imply that
\begin{eqnarray}
m_I^C & \geq & 7^{-\frac{C}{C+2}} m^\frac{C}{C+2}\geq m^{\frac{C-1}{C+2}}=m^{1-\frac{3}{C+2}}\ge m^{1-\varepsilon}
\end{eqnarray}
where the final inequality holds since $C+2\geq C\geq \frac{3}{\varepsilon}$.
\end{proof}

\section{Conclusions and open problems}
\label{sec:conc}
In this paper, we have presented a range of novel polynomial-time solvability and NP-hardness results for variants of \hrc, with a heavy focus on the restriction that the couples' preferences are \resp\ and \subcomp.  The results in this paper contribute to building knowledge about the frontier between polynomial-time solvable and NP-hard variants.

For example, whilst Theorem \ref{thm:respNPc} showed that \hrc\ is NP-hard for general \resp\ and \subcomp\ preferences, Theorem \ref{thm:12abc} showed that if the couples additionally belong to one of several types, \hrc\ is solvable in polynomial time.  This provides one type of boundary, and it remains open to identify other types of couples that could give rise to additional polynomial-time solvable cases.

Whilst Corollary \ref{33comdmml} showed that \hrcdual\ is NP-hard even if every preference list is of length 3 and derived from a master list of single doctors, couples and hospitals, Corollary \ref{cor:hrcdual} showed that \hrcdual\ is solvable in polynomial time for \resp\ and \subcomp\ preferences.

Finally, whilst \cite{biro2014hospitals} showed that \hrc\ is NP-hard even if every couple's preference list is of length at most 2, Corollary \ref{cor:hrclength1} showed that the problem is solvable in polynomial time if every couple's preference list is of length 1. 

An important future direction is to continue to expand our knowledge of the frontier between tractable and intractable variants of \hrc.  A key message from this paper, however, is that \hrc\ \emph{can} be tractable in a range of different scenarios, and perhaps much more so than might have been predicted previously, given the inherent intractability of the problem in general and in so many restricted settings.

%

\section*{Acknowledgements}
Gergely Csáji was supported by the Hungarian Academy of Sciences, Momentum grant number LP2021-2/2021, by the Hungarian Scientific Research Fund, OTKA, grant number K143858 and by the Ministry of Culture and Innovation of Hungary from the National Research, Development and Innovation fund, financed under the KDP-2023 funding scheme (grant number C2258525). David Manlove was supported by the Engineering and Physical Sciences Research Council, grant number EP/P028306/1.  Iain McBride was supported by a Scottish Informatics and Computer Science Alliance Prize Studentship.  James Trimble was supported by the Engineering and Physical Sciences Research Council, Doctoral Training Partnership number EP/N509668/1.

We would like to thank the anonymous reviewers of the preliminary version of this paper submitted to IJCAI 2024 for their valuable comments, which helped to improve the presentation of this paper.

For the purpose of open access, the authors have applied a Creative Commons Attribution (CC BY) licence to any Author Accepted Manuscript version arising from this submission.
\bibliographystyle{plain}
\bibliography{matching}
\end{document}